\documentclass{article}

\usepackage{float}      % Improved interface for floating objects
\usepackage{graphicx}   % Enhanced support for graphics
\usepackage{hyperref}   % Extensive support for hypertext in LaTeX
\usepackage{latexsym}   % A collection of LaTeX symbols
\usepackage{xcolor}     % Driver-independent color extensions for LaTeX
\usepackage{lipsum}     % General Lorem ipsum text
\usepackage{subcaption}
\usepackage{amsthm}
\usepackage{amssymb}    % AMS font symbols
\usepackage{amsmath}    % AMS equation environments
\usepackage{physics}    % For \dd and other math commands
\usepackage{bbm}        % Bitmap version of Computer Modern fonts
\usepackage{eucal}      % Euler script fonts
\usepackage{mathabx}  % provides \widecheck
\usepackage{mathrsfs}   % Rich set of script (calligraphic) fonts
\usepackage{verbatim}

\usepackage{multirow}   % Create tabular cells spanning multiple rows
\usepackage{ulem}       % Package for underlining
\usepackage{dsfont}

\usepackage[numbers,sort&compress]{natbib}  % Bibliography management

\usepackage{enumerate}  % Enumerate with redefinable labels
\usepackage{enumitem}   % Control layout of itemize, enumerate, description
\usepackage{authblk}

\newtheorem{theorem}{Theorem}[section]

\newtheorem{definition}[theorem]{Definition}

\newtheorem{lemma}[theorem]{Lemma}
\newtheorem{example}[theorem]{Example}
\newtheorem{Assumption}{Assumption}
\newtheorem{rem}{Remark}

\renewcommand{\theAssumption}{\Alph{Assumption}}
\makeatletter
\newcommand{\settheoremtag}[1]{% \settheoremtag{<tag>}
	\let\oldtheAssumption\theAssumption% Store \thetheorem
	\renewcommand{\theAssumption}{#1}% Redefine it to a fixed value
	\g@addto@macro\endAssumption{% At \end{theorem}, ...
		\global\let\theAssumption\oldtheAssumption}% ...restore \thetheorem
}
\makeatother

\def\wt{\widetilde}
\def\wh{\widehat}

\newcommand{\1}{\mathbbm{1}} % Indicator function
\def\E{\mathbb{E}} % Shortcut for expectation
\newcommand{\R}{\mathbb{R}}

\newcommand{\bignorm}[1]{{\big\vert\kern-0.25ex\big\vert #1 \big\vert\kern-0.25ex\big\vert}}
\newcommand{\opnorm}[1]{{\vert\kern-0.25ex\vert\kern-0.25ex\vert #1 \vert\kern-0.25ex\vert\kern-0.25ex\vert}}
\def\bbR{\mathbb{R}}

  \def\bfX{{\mathbf X}}

\newcommand{\calC}{\mathcal{C}} 
 \newcommand{\calF}{\mathcal{F}}
 
 \newcommand{\calJ}{\mathcal{J}}
 \newcommand{\calL}{\mathcal{L}}

\newcommand{\calU}{\mathcal{U}}

\newcommand{\red}{\color{red}}

\allowdisplaybreaks

\begin{document}
\setlength{\parindent}{0pt}
\date{\today}

\title
{Trading with market resistance and concave price impact}
\author[1]{Nathan De Carvalho}
\author[2]{Youssef Ouazzani Chahdi}
\author[3]{Gr\'egoire Szymanski}

\affil[1]{LPSM, Université Paris Cité, \texttt{decarvalho@lpsm.paris}}
\affil[2]{MICS, CentraleSupélec,
\texttt{youssef.ouazzani-chahdi@centralesupelec.fr}}
\affil[3]{DMATH, Université du Luxembourg, \texttt{gregoire.szymanski@uni.lu}}

\maketitle

\begin{abstract}
    We consider an optimal trading problem under a market impact model with endogenous market resistance generated by a sophisticated trader who (partially) detects metaorders and trades against them to exploit price overreactions induced by the order flow. The model features a concave transient impact driven by a power-law propagator with a resistance term responding to the trader's rate via a fixed-point equation involving a general resistance function. We derive a (non)linear stochastic Fredholm equation as the first-order optimality condition satisfied by optimal trading strategies. Existence and uniqueness of the optimal control are established when the resistance function is linear, and an existence result is obtained when it is strictly convex using coercivity and weak lower semicontinuity of the associated profit-and-loss functional. We also propose an iterative scheme to solve the nonlinear stochastic Fredholm equation and prove an exponential convergence rate. Numerical experiments confirm this behavior and illustrate optimal round-trip strategies under ``buy'' signals with various decay profiles and different market resistance specifications.
\end{abstract}


%\tableofcontents

% {\red
% \section*{General comments \& todos}

% \begin{enumerate}
%     \item Identification of the main theoretical results + key assumptions to be added in the introduction (Sections mal référencées pour la partie théorique)

%     \item Hypothèses sur le noyau à clarifier

%     \item Abstract to be written down (once all the key results have been identified)

%     \item Choix des preuves à mettre dans le corps principal du texte? un papier de maths devrait contenir les preuves clés concises - les arguments nouveaux / élégants dans le corps du texte, les définitions de contexte / résultats indirects en appendix - qu’en pensez-vous ?

%     \item Choice of the title?
    
%     \item Check the references: titles, date, published or not?

%     \item Bonus question: Que dire de l’existence d’arbitrage-par-round-trip dans notre modèle? Comparer avec les arguments de Jim G.
    
%     \item Check notations are consistent: denote $\mathbf{G}u$ for any linear operator $\mathbf{G}$, else add parentheses (i.e. $\mathbf{G}(u)$ for any nonlinear $\mathbf{G}$)
% \end{enumerate}
% }

\section{Introduction}

A central feature of modern financial markets is \textit{market impact}, that is, the empirically observed positive correlation between the sign of a sizable incoming market order and the subsequent price change. Market impact plays a crucial role in execution costs, risk management, and market design; see, for instance, \citet{almgren2005direct}, \citet{freyre2004review}, \citet{hey2023cost}, \citet{robert2012measuring}, and \citet{bouchaud2018trades}. Modeling this phenomenon is particularly relevant when devising strategies for executing \textit{metaorders}—large transactions placed by institutional traders and implemented through a sequence of smaller child orders over a given time horizon. The presence of market impact has long been recognized as a fundamental feature of market microstructure, and understanding its form is essential from both empirical and theoretical perspectives.\\

Measuring market impact is inherently challenging due to its noisy nature. Statistical studies therefore tend to focus on the execution of metaorders, which induces a persistent liquidity imbalance and generates price moves that can be identified statistically. However, during the execution of any given metaorder, numerous other trades occur simultaneously, contributing additional noise. Careful statistical procedures and averaging over many metaorders help to filter out part of this noise, allowing some universal properties of market impact to emerge; see, for example, \citet{almgren2005direct}, \citet{bacry2015market}, \citet{bershova2013non}, and \citet{bucci2019crossover}. Numerous empirical studies document that prices react mechanically during the execution of a metaorder, exhibiting a concave dependence on traded volume that peaks at the end of the metaorder—a phenomenon known as the \textit{square-root law}—followed by a convex relaxation phase; see \citet{lillo2003master}, \citet{hopman2003essays}, \citet{almgren2005direct}, \citet{bershova2013non}, \citet{gatheral2010no}, Moro et al. \cite{moro2009market}, \citet{bouchaud2018trades}, \citet{kyle2023large}, and \citet{Sato_24}. This square-root dependence, however, primarily holds for large traded volumes, while the impact of small orders is approximately linear in volume, as shown in \citet{bucci2019crossover}. To capture this dual behavior, \citet{benzaquen2018market} proposes the market impact model
\begin{equation*}
MI(Q_t) \approx c \sigma \bigg( \frac{Q_t}{V} \bigg)^{1/2} \mathcal{F}\bigg(\frac{Q_t}{V}\bigg),
\end{equation*}
where $c$ is a constant of order~1, $\sigma$ denotes the daily volatility, $V$ the typical daily traded volume, $Q_t$ the executed volume of the metaorder at time $t$, and $\mathcal{F}$ a monotone function satisfying $\mathcal{F}(x)\approx \sqrt{x}$ as $x\to0$ and $\mathcal{F}(x)\to a$ as $x\to\infty$ for some $a>0$.\\

On the one hand, several theoretical frameworks have been developed to rationalize these empirical findings, ranging from latent order book models, such as those of \citet{toth2011anomalous} and \citet{donier2015fully}, to equilibrium models with strategic traders; see \citet{gabaix2006institutional}. More recently, the impact of limit orders has also been investigated, see \citet{chahdi2024theory}. A notable recent contribution is due to \citet{durin2023two}, who introduce a model with informed traders. The key idea is that certain sophisticated agents can detect the presence of a metaorder and infer its effect on the observed price. Their framework distinguishes between the true price of the asset, which excludes the market's overreaction to the metaorder, and the observed price, which incorporates the mechanical impact of the metaorder. By trading strategically against the detected order flow, informed traders generate a form of \textit{market resistance} that alters the shape of the impact curve. Their analysis leads to two square-root laws of market impact—one in time and one in participation rate—providing a microstructural foundation for well-established empirical results and offering a refined perspective on how informed liquidity provision affects price formation.\\

On the other hand, a substantial body of literature aims to formalize the connection between reduced-form market impact models $I^{X}$ and price features such as signals or volatility, with the goal of deriving trading strategies $X$ that minimize execution costs. In the seminal work of \citet{almgren2005direct}, the execution cost of a trading strategy combines temporary and permanent impact terms that shift the asset price linearly with the total traded volume. The trader’s objective is then to minimize the expected cost and risk of liquidation over a finite horizon, typically expressed through a mean–variance functional. A key refinement concerns the transient nature of impact, whereby the observed asset price is modeled as the convolution of past order flow with a decaying kernel $G$ -- a formulation known as the \textit{linear propagator}, introduced in discrete time by \citet{bouchaud2003fluctuations}. In continuous time, the impacted price $S^X$ is written as
\begin{equation} \label{eq:linear_propagator}
S_t^{X} = S_0 + I_{t}^{X}, \qquad 
I_{t}^{X} := \int_0^t G(t-s)\, \dd X_s, \qquad t \geq 0,
\end{equation}
where $X$ denotes the signed traded volume. The kernel $G$ governs how the impact of past trades decays over time: an exponential decay \`a la \citet{obizhaeva2013optimal} corresponds to short memory, whereas power-law kernels, as in the framework of \citet{abi2022optimal}, capture long-lasting effects and are more consistent with empirical propagator estimates; see \citet{bouchaud2008marketsslowlydigestchanges} and \citet{donier2015fully}. A central concern for such models is ensuring financial well-posedness, meaning that the impact model itself cannot be exploited to generate positive profits. \citet{gatheral2010no} showed that this requirement is equivalent to the positive semi-definiteness of the kernel $G$, yielding the so-called ``no-dynamic-arbitrage'' condition.\\

Building on the optimal execution framework with linear price impact, recent works have formulated increasingly general \textit{optimal trading problems} of the form
\begin{equation} \label{eq:generic_form_optimal_trading}
    \sup_{u \in \mathcal A} \E \bigg[ \int_{0}^T (\alpha_{t} - I_{t}^{X})\, \dd X_{t} + M(X) \bigg],
\end{equation}
where $\alpha$ captures stochastic signals, $M$ encodes soft constraints and risk-aversion penalties, and $\mathcal{A}$ denotes the set of admissible trading strategies. For instance, \citet{abi2024trading} extends the linear propagator framework to incorporate linear functional constraints (e.g., no-shorting or no-buying constraints, stochastic stop-trading rules), demonstrating that such constraints can be handled within a non-Markovian continuous-time model while preserving numerical tractability. To account for the concave nature of price impact, \citet{alfonsi2010optimal} introduces a nonlinear impact function applied to exponentially decaying kernels, extending the model of \citet{obizhaeva2013optimal}, and \citet{hey2025trading} derives explicit optimal inventories in this setting under general price signals. Building on these results, \citet{abi2025fredholm} proposes a general nonlinear propagator model formulated in terms of the trading rate, which naturally encompasses nonlinear impact functions as well as power-law decay, and develops numerical methods to compute optimal strategies.\\

The goal of this paper is to study optimal trading strategies in a market impact setting that incorporates market resistance, following the microstructural insights of \citet{durin2023two}. Leveraging the techniques developed in \citet{abi2025fredholm}, we first establish theoretical results ensuring the well-posedness of the problem, and then derive a first-order optimality condition expressed as a nonlinear stochastic Fredholm equation. This characterization enables the use of numerical methods based on Nyström approximations and Fredholm operator inversion, through which we compute optimal trading strategies in practice.\\

The remainder of the paper is organized as follows. Section~\ref{sec:market_impact_model} introduces market resistance through a simple game-theoretic model and presents the market impact framework of \citet{durin2023two} together with its microstructural foundations. Section~\ref{sec:optimal_trading} formulates the corresponding optimal trading problem, establishes the main existence results, and derives the associated first-order optimality condition as a (non)linear stochastic Fredholm equation. Section~\ref{sec:numerical_applications} then presents the numerical scheme used to solve this equation and illustrates the resulting optimal round trips in the presence of stochastic ``buy'' signals.

\section*{Notations}
We fix a finite time horizon $T > 0$ and a filtered probability space $(\Omega, \mathcal F, (\mathcal F_t)_{t \in [0,T]}, \mathbb{P})$ satisfying the usual conditions. We denote by $\dd t$ the Lebesgue measure on the Borel $\sigma$-algebra $\mathcal{B}([0,T])$, and by $\dd t \otimes \mathbb{P}$ the product measure on the $\sigma$-algebra $\mathcal{B}([0,T]) \otimes \mathcal{F}$. For all $p \geq 1$, we introduce the standard Banach spaces
\begin{equation*} 
    \mathcal{L}^{p} := \bigg\{ f: [0,T] \times \Omega \to \R \ \text{progressively measurable} \; : \; \E \bigg[ \int_{0}^{T}  |f_{t}|^{p} \dd t \bigg] < \infty \bigg\},
\end{equation*}
and, for $p=2$, we equip $\mathcal L^2$ with the inner product
\begin{equation*}
    \langle f, g \rangle := \E \bigg[ \int_{0}^{T}  f_{t} g_{t}\, \dd t \bigg], \quad f,g \in \mathcal{L}^{2},
\end{equation*}
which makes it a Hilbert space with associated norm $\| f\| := \sqrt{\langle f,f\rangle}$. Similarly, we denote by $L^p$, $p \geq 1$, the standard Banach spaces with respect to the Lebesgue measure. For $u \in \mathcal{L}^{1}$ and $t \in [0,T]$, we denote by $\E_t u$ the conditional expectation of $u$ with respect to the $\sigma$-algebra $\mathcal{F}_t$.

\section{Inception: a sophisticated trader as a market resistance}
\label{sec:market_impact_model}

\subsection{A two-player toy model}

Consider a market with a risky asset and two traders, Alice and Bob, observing each other's decisions, with the following motivations:
\begin{enumerate}
    \item Alice wants to trade the risky asset following her market signal $\alpha$,

    \item Bob distrusts Alice’s signal and seeks only to profit by trading against her flow; his orders create a resistance flow that opposes her impact.
\end{enumerate}
Denoting by $u^{A}, u^{B}$ the respective trading rates of Alice and Bob, assume they both suffer trading costs linear in the sum of their trading rates $u := u^{A} + u^{B}$, following the linear propagator model \eqref{eq:linear_propagator} with instantaneous slippage as in~\cite{abi2022optimal}. Formally, we write
\begin{equation*}
I^{u} := \frac{1}{2} u + \mathbf{G}u,
\end{equation*}
where $I^u$ is the aggregated market impact from both Alice's and Bob's trades, and $(\mathbf{G}u)_{t} := \int_{0}^{t} G(t,s) u_{s} \dd s$ captures the transient nature of market impact via the kernel $G$. Their respective PnL functionals are given by
\begin{align*}
    \mathcal{J}^{A}(u^{A}) & := \mathbb{E} \bigg[ \int_{0}^{T} \big( \alpha^{A}_t(u^{B}) - I_{t}^{u^{A}} \big) u_{t}^{A} \dd t  \bigg], \\
    \mathcal{J}^{B}(u^{B}) & := \mathbb{E} \bigg[ \int_{0}^{T} \big( \alpha^{B}_t(u^{A}) - I_{t}^{u^{B}} \big) u_{t}^{B} \dd t \bigg],
\end{align*}
where Alice's and Bob's respective effective signals $\alpha^{A}, \alpha^{B}$ are given by
\begin{align*}
    \alpha^{A}(u^{B}) := \alpha - \frac{1}{2} u^{B} - \mathbf{G}u^{B},
    \qquad\text{ and }\qquad
    \alpha^{B}(u^{A}) := - \frac{1}{2} u^{A} - \mathbf{G}u^{A}.
\end{align*}
The main question is: can we find a couple of strategies $(\hat{u}^{A}, \hat{u}^{B})$ that maximize $\mathcal{J}^{A}$ and $\mathcal{J}^{B}$ respectively?\\

Applying~\cite[Proposition 5.1]{jaber2023equilibrium}, the first-order conditions for Alice's and Bob's trading problems write respectively
\begin{align*}
    u^{A} = ( \mathbf{id} - \mathbf{B} )^{-1} a^{\alpha^{A}(u^{B})},     
    \qquad\text{ and }\qquad
    u^{B} = ( \mathbf{id} - \mathbf{B} )^{-1} a^{\alpha^{B}(u^{A})},
\end{align*}
where the operator $\mathbf{B}$ as well as the processes $a^{\alpha^{A}(u^{B})}$, $a^{\alpha^{B}(u^{A})}$ are explicitly given by
\begin{align*}
    a_{t}^{\alpha} & := \alpha_{t} - \langle \mathds{1}_{\{t<.\}} G(.,t), \mathbf{D}_{t}^{-1} \mathds{1}_{\{t<.\}} \E_t [ \alpha_{.} ] \rangle_{L^2}, \quad \alpha \in \mathcal{L}^{2} \\
    B(t,s) & := \mathds{1}_{\{s<t\}} ( \langle \mathds{1}_{\{t<.\}} G(.,t), \mathbf{D}_{t}^{-1} \mathds{1}_{\{t<.\}} G(.,s) \rangle_{L^2} - G(t,s) ), \\
    \mathbf{D}_{t} & := \mathbf{id} + \mathbf{G}_{t} + (\mathbf{G}_{t})^{*}, \quad G_{t}(s,r) := G(s,r) \mathds{1}_{r \geq t}, \quad t \in [0,T].
\end{align*}
Now injecting the definition of $u^{A}$ into that of $u^{B}$ yields the following fixed-point equation for Bob's trading rate $u^{B}$
\begin{equation} \label{eq:Bob_fixed_point}
    u^{B} = ( \mathbf{id} - \mathbf{B} )^{-1} a^{\alpha^{B}(( \mathbf{id} - \mathbf{B} )^{-1} a^{\alpha^{A}(u^{B})})}.
\end{equation}
Therefore, if there is a $\hat{u}^{B}$ that satisfies~\eqref{eq:Bob_fixed_point}, setting
\begin{equation*}
    \hat{u}^{A} = ( \mathbf{id} - \mathbf{B} )^{-1} a^{\alpha^{A}(\hat{u}^{B})}
\end{equation*}
yields a couple $(\hat{u}^{A}, \hat{u}^{B})$ that maximizes $\mathcal{J}^{A}$ and $\mathcal{J}^{B}$ respectively. Finally, denoting by $\hat{r}^{B} := -\hat{u}^{B}$ the resistance induced by Bob's trades, Alice's optimal PnL rewrites as follows
\begin{align*}
    \mathcal{J}^{A}(\hat{u}^{A}) & = \mathbb{E} \bigg[ \int_{0}^{T} \big( \alpha - \frac{1}{2} ( \hat{u}^{A} - \hat{r}^{B} ) - \mathbf{G} (\hat{u}^{A} - \hat{r}^{B} ) \big) \hat{u}_{t}^{A} \dd t \bigg], \\
    \hat{r}^{B} & = - ( \mathbf{id} - \mathbf{B} )^{-1} a^{\alpha^{B}\big(( \mathbf{id} - \mathbf{B} )^{-1} a^{\alpha^{A}(-\hat{r}^{B})}\big)},
\end{align*}
and this simple two-player game setting indeed motivates the introduction of a market resistance caused by the presence of an informed trader, and solution to a fixed-point equation.

\subsection{A general market resistance model}

Starting again from the linear propagator framework \eqref{eq:linear_propagator}, consider the following impact kernel specification
\begin{equation}
\label{eq:constant_and_power_law_kernel}
G(t) \;:=\; \kappa_\infty + G_{\lambda, \nu}(t), \quad t > 0,
\end{equation}
which combines a permanent component proportional to $\kappa_{\infty} \geq 0$ capturing the trader's contribution to the long-term fundamental value of the asset, and a power-law decaying kernel
\begin{equation} \label{eq:decay_specification}
    G_{\lambda, \nu}(t) := \lambda t^{\nu-1}, \quad \lambda > 0, \quad \nu \in \bigg( \frac{1}{2}, 1 \bigg),
\end{equation}
reproducing the transient long-memory price impact of trades empirically observed in order flow, see Chapter 13.2.1 from \cite{bouchaud2018trades} and references therein. $G$ is positive semi-definite, which ensures the absence of dynamic price manipulation and rules out profitable round trips \cite{gatheral2010no}. Although we focus on specification \eqref{eq:constant_and_power_law_kernel} for concreteness and interpretation, the framework extends naturally to other admissible kernels, such as completely monotone kernels or sums of exponentials, which have also been proposed in the literature to capture different impact decay patterns, see for instance \cite{abi2025fredholm}.\\

%{\red (this needs to be carefully backed up, some articles are not in favor of permanent price impact, see for instance \cite{bouchaud2009price}[Section 4.2])}

A key implication is that an impact of the form $\int_0^t G(t-s)\, \dd X_s$, with $G$ given by \eqref{eq:constant_and_power_law_kernel}, overreacts at short horizons: for any finite $t > 0$, one has $G(t) > \lim_{s \to \infty} G(s) = \kappa_\infty$. In economic terms, prices initially move more than what is justified by the long-run information content of the trader's order flow. This transient overshooting creates a temporary mispricing between the observed market price and the trader's contribution to the long-term efficient long-term price, which is considered a temporary price signal denoted by $\alpha^{r}$ such that
\begin{equation}
    \alpha_{t}^{r} := \int_0^t \big( G(t-s) - G(\infty) \big)\big(u_s - r^u_s\big)\,\dd s = \int_0^t G_{\lambda, \nu}(t-s)\big(u_s - r^u_s\big)\,\dd s.
\end{equation}
In classical linear propagator models, $\alpha^{r}$ is purely mechanical and is not acted upon. A key contribution of \citet{durin2023two} is to formulate an endogenous correction to such mispricing caused by the trader's strategy $u$ by introducing a continuous-time market resistance $r^{u}$, generated by sophisticated traders. These agents partially detect the presence of the trader's metaorder, infer the transient component of its impact \cite{vaglica2008scaling,toth2010segmentation}, and trade against it to benefit from $\alpha^{r}$. Their activity produces an opposing flow that dampens impact dynamically, rather than
letting prices passively relax through the kernel $G$. In their framework, the pathwise market impact of a metaorder with execution strategy $u$ at date $t \geq 0$, meaning that $X_t = \int_0^t u_s\,\dd s$, is given by
\begin{equation}
\label{eq:market_impact_rewrite}
MI(u,t) = \int_0^t G(t-s)\big(u_s - r^u_s\big)\,\dd s ,
\end{equation}
where $r^u$ represents the endogenous reaction rate of sophisticated traders. This reaction satisfies the following fixed-point equation
\begin{equation}
\label{eq:resistance_rewrite}
r^u_t = \mathcal{U}\!\left( \alpha_{t}^{r} \right), \quad i.e., \; r^u_t = \mathcal{U}\!\left( \int_0^t G_{\lambda, \nu}(t-s)\big(u_s - r^u_s\big)\,\dd s \right),
\end{equation}
for some increasing function $\mathcal{U}:\mathbb{R}\to\mathbb{R}$. A natural specification for $\mathcal{U}$ can be derived heuristically. Assume that price impact follows a square-root law, that is, if the sophisticated trader invests an amount $x$ proportional to this signal $\alpha^{r}$, the expected gain is $\alpha^{r} x$, while the associated impact cost is $k \sqrt{x}\,x$ for some $k>0$. The expected profit, $\alpha^{r} x - k \sqrt{x}\,x$, is maximized for $x = \tilde{k}\,\left(\alpha^{r}\right)^2$ with $\tilde{k}>0$, which motivates taking $\mathcal{U}(x)$ in \eqref{eq:resistance_rewrite} proportional to $x^2$. This particular choice is not required in the remainder of the paper.\\

Importantly, the resistance mechanism is backed by a precise microstructural foundation where it can be derived from a high-frequency model in which market orders follow self-exciting point processes and prices are set as conditional expectations of future order flow. When informed agents partially filter out the metaorder component and trade against the resulting mispricing, the scaling limit of the model leads exactly to \eqref{eq:market_impact_rewrite}--\eqref{eq:resistance_rewrite}. We refer to \citet{durin2023two} for the full derivation. 
%While in \citet{durin2023two} the kernel is necessarily of the form \eqref{eq:constant_and_power_law_kernel}, subsequent work shows that, following \citet{abi2021weak} and \citet{szymanski2025mean}, any completely monotone kernel can be obtained through the same mechanism {\red (what do you mean exactly here?)}.\\

Economically, $r^u$ can be interpreted as a stabilizing force. When a metaorder pushes prices above their efficient level, sophisticated traders increase their activity, absorbing part of the order flow and limiting further price pressure. When the mispricing shrinks, their activity naturally subsides. The function $\mathcal{U}$ governs how aggressively the resistance reacts to perceived mispricing, as illustrated later on in Figure~\ref{F:trading_buy_signals_various_convexity}.\\

\begin{comment}
{\red (This paragraph seems off-topic with respect to the the Subsection's title - if you want to refer to latent liquidity models, it would be more natural to do so in the introduction - what do you think?)} This formulation is fully consistent with Bouchaud’s theory of latent liquidity introduced in \citet{toth2011anomalous} and \citet{mastromatteo2014agent}. In this view, liquidity is not uniform across prices but exhibits a V-shaped profile that vanishes near the current price and increases away from it, reflecting the delayed revelation of hidden supply and demand \cite{benzaquen2018market}. Impact then arises from the progressive activation of latent orders as prices move, leading to a nonlinear relation between executed volume, participation rate, and price response. The hydrodynamic limit of the latent order book yields a unified expression interpolating smoothly between linear and square-root impact regimes \cite{donier2015fully,benzaquen2018market,bucci2019crossover}. More generally equilibrium limit-order-book models show that power-law impact can also emerge endogenously from optimal trading against liquidity suppliers when asset returns exhibit heavy tails \cite{ccetin2024power}.
\end{comment}

\subsection{Properties of the market resistance}

The introduction of a resistance term $r^u$ satisfying \eqref{eq:resistance_rewrite} fundamentally alters the shape of market impact. Assuming a buy strategy $u \geq 0$ with compact support, the market impact $MI$ from \eqref{eq:market_impact_rewrite} is decomposed into a permanent and a transient components, denoted respectively by $PMI$ and $TMI$, such that
\begin{align*}
&PMI(u) := \lim_{t \to \infty} MI(u,t) = \kappa_\infty \int_0^\infty \big(u_{s}-r^u_s\big)\,\dd s ,
\\
&TMI(u,t) := MI(u,t) - PMI(u) = \int_0^t G_{\lambda, \nu}(t-s)\big(u_{s}-r^u_s\big)\,\dd s
-\kappa_\infty \int_t^\infty \big(u_{s}-r^u_s\big)\,\dd s .
\end{align*}
As in standard propagator models, impact is concave in time and decays as a power law with exponent $\nu-1$; see \cite[Section~13.4.4]{bouchaud2018trades} and Figure~\ref{fig:mi_in_time}. 

\begin{figure}[H]
    \centering

    \begin{subfigure}{\textwidth}
        \centering
        \includegraphics[width=4in]{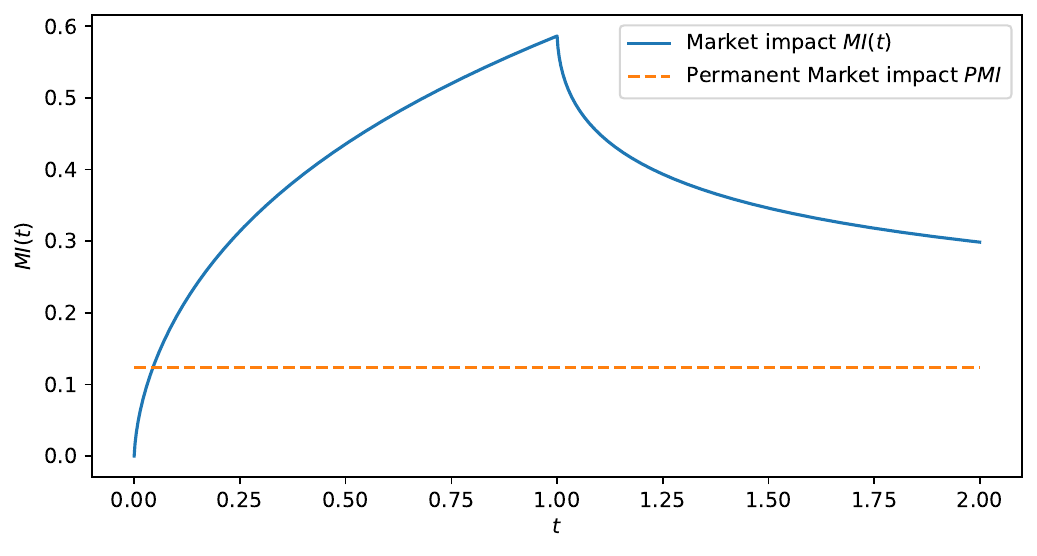}
        \vspace{-0.3cm}
        \caption{Market impact profile}
        \label{fig:mi_only}
    \end{subfigure}

    \vspace{0.2cm}

    \begin{subfigure}{\textwidth}
        \centering
        \includegraphics[width=4in]{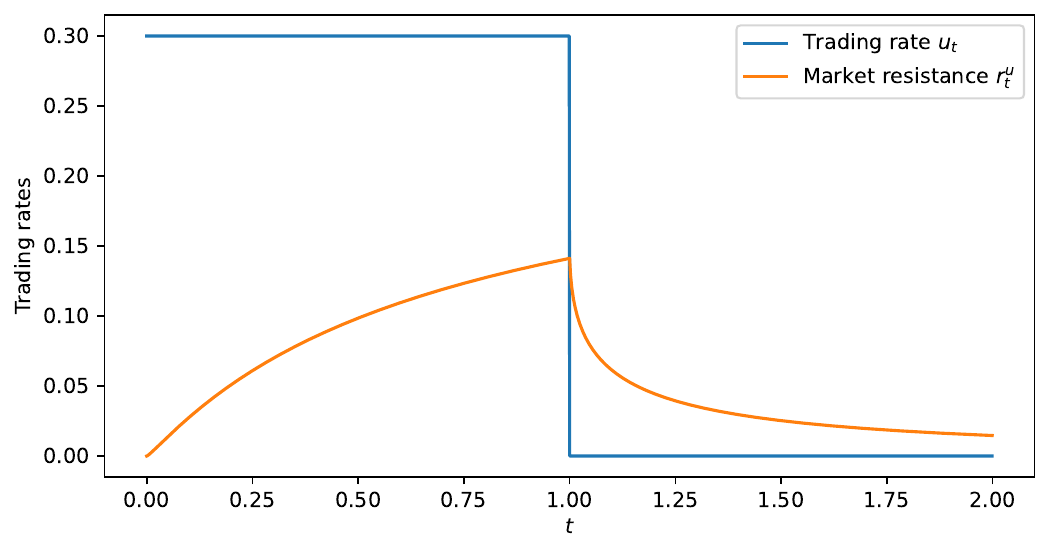}
        \vspace{-0.3cm}
        \caption{Trading rate $u$ and its resulting endogenous resistance rate $r^u$}
        \label{fig:execution}
    \end{subfigure}

    \caption{Market impact and trading rates for $\nu = 0.5$, $\lambda = 1$,
    $u(t) = 0.3\,\mathbf{1}_{0 \leq t \leq 1}$, and $\mathcal{U}(x)=x^2$.}
    \label{fig:mi_in_time}
\end{figure}

%\begin{figure}[H]
%    \centering
%    \includegraphics[width=6in]{Figures/MI_execution.pdf}
%    \caption{%
%    \textbf{Left:} Market impact profile $MI(t)$ over time together with its permanent component $PMI$.
%    \textbf{Right:} Trading rate $u_t$ and resulting execution intensity $r^u_t$.
%    Parameters are $\nu = 0.5$, $\lambda = 1$, 
%    $u(t) = 0.3\,\mathbf{1}_{0 \leq t \leq 1}$, and $\mathcal{U}(x)=x^2$.}
%    \label{fig:mi_in_time}
%\end{figure}

A particularly important implication concerns the dependence of impact on the participation rate. Consider a normalized execution profile $u$ (i.e., $\int_0^\infty u(s)\,\dd s = 1$) and suppose trading occurs at rate $\gamma u$. In this setting, $\gamma$ has no unit, and acts as a proxy for the participation rate, which is $\gamma/(\gamma+V)$, where $V$ denotes the typical background market volume. Under the power-law specification $\mathcal{U}(x)=c x^c$ with $c>1$, \citet[Theorem~9]{durin2023two} show that market impact satisfies
\[
MI(\gamma u,t)\sim_{\gamma\to\infty} \gamma^{1/c}.
\]
Choosing $c=2$ recovers the empirically observed square-root law. In this interpretation, the square-root scaling is not imposed exogenously but emerges from the interaction between aggressive execution and endogenous market resistance. Although this theorem is asymptotic and valid for large participation rates, Figure~\ref{fig:mi_in_pr} illustrates that the scaling remains a good approximation even for moderate values of $\gamma$.

\begin{figure}[H]
    \centering
    \includegraphics[width=4in]{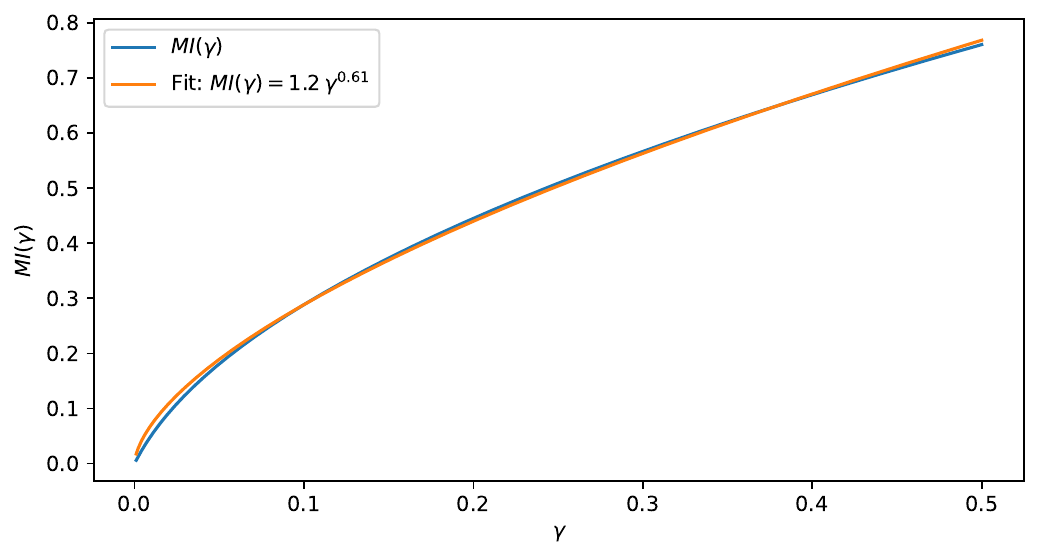}
    \caption{Market impact as a function of $\gamma$ for $\nu=0.5$,
    $\lambda=1$, and $\mathcal{U}(x)=x^2$. A power-law fit of the form
    $MI = 1.172 \cdot \gamma^{0.6086}$ is shown.}
    \label{fig:mi_in_pr}
\end{figure}

\section{Optimal trading with market resistance} \label{sec:optimal_trading}
\subsection{Problem formulation}

Consider a risky asset whose unaffected price process $S$ is given by
\begin{equation} \label{eq:unaffected_price_process}
    S_t = P_t + M_t, \quad t \in [0,T],
\end{equation}
where $P$ a finite variation process valued in $\mathcal L^2$ and $M$ is a centered square-integrable martingale. We consider an agent with an initial holding $X_0 \in \bbR$ of this asset controlling his trading rate $u \in \calL^2$ of his inventory $X^u$ such that
\begin{equation}
\label{eq:inventory}
    X_t^u := X_0 + \int_0^t u_s \, \dd s, \quad t \in [0,T].
\end{equation}
In particular,
\begin{equation} \label{eq:finite_sup_squared_holding}
    \sup _{t \in[0, T]} \mathbb{E}[|X_t^u|^2]<\infty.
\end{equation}
Furthermore, trading the asset at the rate $u$ impacts the unaffected price $S$ such that the effective trading price $S^{u}$ becomes
\begin{equation} \label{eq:effective_traded_price}
    S^{u} := S + I^{u},
\end{equation}
where the price impact model $I^u$ is specified as follows
\begin{equation} \label{eq:price_impact_model}
    I_t^u := \frac{\gamma}{2} u_{t} + (\mathbf{H} + \mathbf{G} )(u-r^{u})_{t}, \quad t \in [0,T].
\end{equation}
Here, $\mathbf{H}$ and $\mathbf{G}$ are two operators on $\mathcal{L}^2$ modeling the permanent and transient market impacts, respectively, while $\gamma \geq 0$ denotes, from this point onward, the (constant) slippage cost intensity. As motivated in Section~\ref{sec:market_impact_model}, $r^{u}$ denotes the \textit{market's reaction} or \textit{resistance} to the trading strategy $u$, capturing the nonlinear nature of price impact as it solves the following nonlinear Volterra equation
\begin{equation}
\label{eq:r_u}
    r^u_t = \calU \big( \mathbf{G}(u-r^{u})_{t}\big), \quad t \in [0,T],
\end{equation}
where the $\calU$ is called \textit{the resistance function}. In order to ensure the well-posedness of the price impact model~\eqref{eq:price_impact_model}--\eqref{eq:r_u}, we make the following assumptions on the operators $\mathbf{H}$, $\mathbf{G}$ and on the resistance function $\calU$.

\begin{definition}
    A Volterra kernel $G \colon [0, T]^2 \to \mathbb{R}$ is said to be \emph{admissible} if
    \begin{equation} \label{eq:constant_norm_of_G}
        C_{G} :=
        \sup_{t\in [0, T]} \int_{0}^{t}|G(t,s)|^2\dd s < \infty.
    \end{equation}
    Any admissible kernel $G$ induces a unique linear and bounded integral operator $\mathbf{G}: \mathcal{L}^{2} \mapsto \mathcal{L}^{2}$ defined by
    \begin{equation*}
        (\mathbf{G}u)_{t} := \int_{0}^{t} G(t,s) u_{s} \dd s, \quad t \in [0,T], \quad u \in \mathcal{L}^{2},
    \end{equation*}
    such that, by Cauchy-Schwarz's inequality, 
    \begin{equation} \label{eq:estimate_on_admissible_G_operator}
        \| \mathbf{G}u \|^2 \le T C_{G} \| u \|^2, \quad u \in \mathcal{L}^2.
    \end{equation}
    As a consequence of Fubini's theorem and the tower property of the conditional expectation, the unique linear and bounded adjoint operator $\mathbf{G}^{\ast}: \mathcal{L}^{2} \to \mathcal{L}^{2}$ of $\mathbf{G}$ is explicitly given by
    \begin{equation*} 
        (\mathbf{G}^{\ast}u)_{t} := \int_{t}^{T} G(s,t) \E_{t}[ u_{s} ]\dd s, \quad t \in [0,T], \quad u \in \mathcal{L}^{2}.
    \end{equation*}
    The boundedness of $\mathbf{G}^\ast$ is ensured by conditional Jensen's inequality and Fubini's theorem such that
    \begin{equation} \label{eq:estimate_on_adjoint}
        \| \mathbf{G}^{\ast}u \|^2 \leq T C_{G} \| u \|^2,\quad u \in \mathcal{L}^2,
    \end{equation}
    see for example \cite[Chapter 6, Section 2]{dunford1988linear}.
\end{definition}
\begin{Assumption}
\label{assumption:G}
    The operator $\mathbf{G}$ is an admissible convolution operator: there exists an admissible kernel function $G: [0, \infty) \to \bbR$ such that for each $u \in \mathcal{L}^2$ and each $t \in [0,T]$, we have
    \begin{equation*}
        \mathbf{G} u_t = \int_0^t G(t-s) u_s \, \dd s.
    \end{equation*}
    We also assume $G$ is continuous on $(0, \infty)$, completely monotone (see Definition~\ref{def:cm}), and $L^2$ integrable on $[0, T]$ and we write
    \begin{equation*}
        \|G\|_{L^2([0,T])}^2 = \int_0^T G(s)^2 \, \dd s.
    \end{equation*}
    Finally, we assume that $G$ satisfies the following continuity assumption
    \begin{equation*}
        \lim\limits_{t \to t_0} \int_0^{\infty} |G(t_0 - s) - G(t-s)| \, \dd s \to 0.
    \end{equation*}
\end{Assumption}

\begin{Assumption}
\label{assumption:H}
    There exists a nonnegative constant $\kappa_{\infty}$ such that for each $u \in \mathcal{L}^2$ and each $t \in [0,T]$, 
    \begin{equation} \label{eq:permanent_impact_specification}
        \mathbf{H} u_t = \kappa_{\infty} \int_0^t u_s \, \dd s.
    \end{equation}
\end{Assumption}

\begin{Assumption}
\label{assumption:U}
    The function $\calU$ is  $L$-Lipschitz continuous with $\calU(0) = 0$ and satisfies
    \begin{equation*}
        \lim\limits_{|x| \to \infty} \frac{\calU(x)}{x} = \delta \geq 0.
    \end{equation*}
\end{Assumption}
\begin{rem}
    Assumption~\ref{assumption:U} implies that $\calU$ is asymptotically close to linear. In particular, it includes the linear case, the bounded case and all cases of interest discussed in Section~\ref{sec:general_case}. Moreover, the assumption that $\calU(0) = 0$ is not essential but is considered to simplify the computations.
\end{rem}

Given any trading strategy $u \in \mathcal{L}^{2}$, the following lemma ensures that the resistance fixed-point equation~\eqref{eq:r_u} admits a unique solution $r^u \in \mathcal{L}^2$ under a linear growth assumption on $\calU$.
\begin{lemma} 
\label{lem:well_posed}
    Let $u \in \calL^2$. Suppose that Assumptions~\ref{assumption:G} and~\ref{assumption:H} hold and that $\calU$ is Lipschitz continuous, then there exists a unique solution $r^{u} \in \mathcal{L}^2$ to~\eqref{eq:r_u}. Moreover, the mapping $u \mapsto r^u$ is Lipschitz continuous on $\mathcal{L}^2$.
\end{lemma}
\begin{proof}
    See Appendix~\ref{app:proof:lem:well_posed}.
\end{proof}
Lemma~\ref{lem:well_posed} guarantees that the price impact $I^{u}$ from~\eqref{eq:price_impact_model} as well as the execution price $(S^u_t)_{t \in [0,T]}$ from~\eqref{eq:effective_traded_price} are well-defined $\mathcal{L}^{2}$ processes, for every $u \in \calL^2$. We are now in place to introduce the agent's Profit and Loss (PnL) functional
\begin{equation}
\label{eq:def:J}
    \calJ(u):=\E\bigg[-\int_0^T S_t^u u_t \dd t+X_T^u S_T-\frac{\phi}{2} \int_0^T(X_t^u)^2 \dd t-\frac{\varrho}{2}(X_T^u)^2\bigg], \quad u \in \mathcal{L}^2,
\end{equation}
where the first term represents the profit and loss induced by the trading rate $u$, the second one values the terminal inventory against the terminal unaffected price \cite[Remark 2.1]{abi2024trading}, the third one encodes risk aversion with weight $\phi \ge 0$, and the last one penalizes the terminal inventory with weight $\varrho \ge 0$. Note that $\calJ$ is indeed well-defined under Assumptions~\ref{assumption:H},~\ref{assumption:G},~\ref{assumption:U} and~\eqref{eq:finite_sup_squared_holding}.  Our objective is to find trading strategies $\hat{u} \in \mathcal{L}^2$ that maximize the performance functional $\mathcal{J}$ such that
\begin{equation} \label{eq:optimal_trading_problem}
\mathcal{J}(\hat{u})=\sup _{u \in \mathcal{L}^2} \mathcal{J}(u).
\end{equation}
Inserting the definitions of the final inventory $X_T^u$ from~\eqref{eq:inventory} and the effective price $S^u$ from~\eqref{eq:effective_traded_price} into $\mathcal{J}$ from~\eqref{eq:def:J}, while using the tower property of the conditional expectation, allows us to rewrite the functional $\calJ$ in the standard form~\eqref{eq:generic_form_optimal_trading} such that
\begin{equation}
\label{eq:J_final}
\mathcal{J}(u)=\E\bigg[\int_0^T(\alpha_t-I_t^u) u_t \dd t + M(X^{u})\bigg]+X_0\E[S_T], \quad u \in \mathcal{L}^2,
\end{equation}
where $\alpha$ is the alpha-signal from the drift of the fundamental asset price $S$, defined by
\begin{equation} \label{eq:alpha_drift}
    \alpha_t = \E_t[S_T - S_t] = \E_t[P_T - P_t], \quad t \in [0,T],
\end{equation}
and
\begin{equation*}
    M(X^{u}) := -\frac{\phi}{2} \int_0^T(X_t^u)^2 \dd t-\frac{\varrho}{2}(X_T^u)^2,
\end{equation*}
in a similar way to \citet{abi2024trading, abi2025fredholm} and \citet{decarvalho2025thesis}.
\begin{rem}
A natural question concerns the choice of the window length $T$. In a real-time trading setup, how should this model be applied, and what value of $T$ is appropriate?

In practice, trading firms typically rely on statistical alpha signals computed over short horizons. For example, suppose an agent trades every 10 minutes using a 10-minute alpha on a given stock. Solving the associated optimal trading problem requires, in our framework, a discrete iterative scheme that takes as input a known trajectory of the buy/sell signal. This assumption is somewhat unrealistic, since the agent only observes a single alpha value every 10 minutes rather than a continuous signal.

How, then, can this framework be used with a single discrete alpha observation? One approach is to solve a sequence of optimization problems, each over a 10-minute window (or, more generally, at the rebalancing frequency of the strategy). Within each window, one may reconstruct a discrete trajectory for the signal---either by recomputing the statistical alpha on smaller subwindows, or by treating the observed alpha as constant over these subintervals---and then apply the numerical scheme described in Section~\ref{sec:numerical_applications}. This yields a practical real-time implementation consistent with the statistical nature of alpha signals.

\end{rem}
% This expression of $\mathcal{J}$ features the standard tradeoff between the exploitation of the asset's alphasignal,

% $$
% \alpha_t:=\mathbb{E}_t[S_T-S_t], \quad t \in [0,T],
% $$
% {\nathan where the first term is the valuation of the final inventory against the terminal unaffected price and the second one is the profit and loss generated from trading with strategy $u$.} Using that $M_T$ is a centered martingale and that $u$ is adapted, this definition can be readily into the standard PnL functional
% \begin{equation*}
% \mathcal{J}(u) =
%     \mathbb{E}\bigg[
%         \int_0^T (\alpha_s - I^u_s) \, u_s d s + X_0 S_T
%     \bigg],
% \end{equation*}

% {\nathan Simplifying the addend $X_0 S_T$, we may as well enforce soft-constraints for the running inventory as well as for the final inventory similarly to~\cite[Equation 2.8]{abi2025fredholm} by adding two respective penalties scaled respectively by $\phi, \; \varrho \geq 0$ such that the gain functional becomes
% \begin{equation} \label{eq:gain_functional}
% \mathcal{J}(u) =
%     \mathbb{E}\bigg[
%         \int_0^T (\alpha_s - I^u_s) \, u_s d s - \frac{\phi}{2} \int_{0}^{T}(X^u_t)^2 d t - \frac{\varrho}{2} (X^u_T)^2
%     \bigg].
% \end{equation}
% }

\paragraph{Reformulation with operators on $\mathcal{L}^2$.} To better analyze the problem, we recast it in operator formulation.
First, note that in view of Lemma~\ref{lem:well_posed}, we know that there exists an operator $\mathbf{R}: \mathcal{L}^2 \to \mathcal{L}^2$ such that $r^u = \mathbf{R} (u), \; u \in \mathcal{L}^{2}$. Using this notation, the market impact $I_t^u$ can be rewritten as
\begin{equation*}
    I_t^u = \frac{\gamma}{2} u_{t} + \big((\mathbf{H} + \mathbf{G}) \circ (\mathbf{I} - \mathbf{R}) u \big)_{t} , \quad t \in [0,T]
\end{equation*}
where $\mathbf{I}$ stands for the identity operator in $\calL^2$ such that $\mathbf{I}u_t = u_t$ for every $u \in \calL^2$ and $t \in [0,T]$. To use the same operator notations, we also write $\bfX u = X^u$ which defines an operator in $\calL^2$. Thus, the gain functional~\eqref{eq:J_final} becomes
\begin{equation}
\label{eq:operator:J}
\mathcal{J}(u) 
    = 
    \langle u, \alpha \rangle - \frac{\gamma}{2} \norm{u}^2 - \langle u, (\mathbf{H} + \mathbf{G})\circ(\mathbf{I} - \mathbf{R})u \rangle
    - \frac{\phi}{2} \norm{\bfX u}^{2}
    - \frac{\varrho}{2} \E[(\bfX u)_T^2] + X_0 \mathbb{E}[S_T].
\end{equation}
All these operators act on $\mathcal{L}^2$. However, we can also see them as operators on $L^2$. In the following, we consider either of each representation depending of the context.

% {\red COMMENT the relationship with Jim Gatheral's Model
% $$
% I^{u} := \mathbf{G}f(u)
% $$
% with
% $$
% f: \mathbb{R} \to \mathbb{R}.
% $$
% In our case, $f:=\mathbf{B}^{-1}$ is a functional. ADD REF + Comment on round-trip.
% }
\subsection{The linear case}
\label{sec:linear_case}
In this section, we assume $\calU$ to be linear. 
We focus on this case for two reasons. First, the operator $\mathbf{R}$ becomes linear, which renders the analysis of $\mathcal{J}$ straightforward: $\mathcal{J}$ becomes a quadratic form, and the existence and uniqueness of a maximizer follows from the positive semi-definiteness of the operator $(\mathbf{H}+\mathbf{G})\circ(\mathbf{I}-\mathbf{R})$. Second, the linear case serves as a building block for the proofs in the general setting. We now formalize the main results of this section in the following statement.
\begin{theorem}
    \label{thm:linear_case}
    Suppose that $\calU(x) = a x$ for some $a \geq 0$ for every $x \in \mathbb{R}$. Suppose also that Assumptions~\ref{assumption:G} and~\ref{assumption:H} hold. Then the functional $- \mathcal{J}$ is $\gamma$-convex, coercive and therefore $\mathcal{J}$ admits a unique maximizer $\wh u$ characterized by 
    \begin{equation}
    \label{eq:FOC:linear}
        \gamma \wh u + (\mathbf{H} + \mathbf{G}) \circ (\mathbf{I} + a \mathbf{G})^{-1} \wh u +  (\mathbf{I} + a \mathbf{G}^*)^{-1} \circ (\mathbf{H}^* + \mathbf{G}^*) \wh u + \mathbf{H}_{\phi,\varrho}\wh u 
    + \mathbf{H}_{\phi,\varrho}^{*} \wh u = \alpha  - X_0 \big(\phi(T- .) + \varrho\big),
    \end{equation}
    where $\mathbf{H}_{\phi,\varrho}$ is the integral operator associated to the kernel
    \begin{equation}
    \label{eq:def:H}
        H_{\phi,\varrho}(t,s) := \big(\phi(T-t)+\varrho\big) \mathbbm{1}_{\{ t > s \}}, \quad s,t \in [0,T].
    \end{equation}
\end{theorem}
The proof is postponed to Appendix~\ref{s:proof_linear_case}.

% {\color{red} Compare here with the results of Jim Gatheral where $a=0$.}

\subsection{The general case}
\label{sec:general_case}
We now aim to extend the results of Section~\ref{sec:linear_case} to the nonlinear case of $\calU$, relying only on Assumption~\ref{assumption:U}. To do this, we first show that $\mathcal{J}$ is Fr\'echet differentiable (in the sense of Definition~\ref{def:frechet:diff}) and we compute its derivative, which allows to identify a first-order condition to be satisfied by any critical point of $\mathcal{J}$. The main difficulty in this step lies in the nonlinearity of $r^u$. We then study the coercivity of $\mathcal{J}$, which ensures that an extremum must exist.

\begin{lemma}
\label{lem:frechet:R}
Suppose that $\calU$ is differentiable and Lipschitz continuous. Then the operator $\mathbf{R}$ is Fr\'echet differentiable in $\mathcal{L}^2$. Moreover, for $u, h$ in $\mathcal{L}^2$, the directional derivative of $\mathbf{R}$ at $u$ in the direction $h$ is the $\mathcal{L}^2$-process $y = D\mathbf{R}(u)(h)$ solution of the Volterra equation
\begin{equation*}
    y_{t} =  \calU' \big(\mathbf{G}(u-r^u)_{t}\big)\mathbf{G}(h - y)_{t}, \quad t \in [0,T].
\end{equation*}
In other words, we have
\begin{align*}
    D\mathbf{R}(u) 
    &= \bigg(\mathbf{I} + \calU' \Big(\mathbf{G}\big(u-\mathbf{R}(u)\big)\Big)\mathbf{G}\bigg)^{-1} \calU' \Big(\mathbf{G}\big(u-\mathbf{R}(u)\big)\Big)\mathbf{G}
    \\
    &= \mathbf{I} - \bigg(\mathbf{I} + \calU' \Big(\mathbf{G}\big(u-\mathbf{R}(u)\big)\Big)\mathbf{G}\bigg)^{-1}.
\end{align*}
\end{lemma}
\begin{proof}
    See Appendix~\ref{ss:proof_lemma_differentiability}.
\end{proof}

Therefore Lemma~\ref{lem:frechet:R} ensures that $\mathcal{J}$ is also Fr\'echet differentiable on $\mathcal{L}^2$. Similar computations to~\cite[Lemma 4.2]{abi2025fredholm} ensure that the Gâteaux derivative of $\calJ$ is given by 
\begin{align} \notag
\nabla\mathcal{J}(u)
    = & \alpha - X_0 \big(\phi(T- .) + \varrho\big) - \gamma u - (\mathbf{H} + \mathbf{G})\circ (\mathbf{I} - \mathbf{R})(u) \\ \label{eq:gateaux_derivative}
    & - \Big(\big(\mathbf{H} + \mathbf{G}\big) \circ \big(\mathbf{I} - D\mathbf{R}(u)\big)\Big)^* u - (\mathbf{H}_{\phi, \varrho} + \mathbf{H}^*_{\phi, \varrho})u,
\end{align}
where $\mathbf{H}^*_{\phi, \varrho}$ is defined in~\eqref{eq:def:H}. We then directly get the following result.

\begin{theorem}
    Any local extrema $u$ of $\mathcal{J}$ must satisfy
    \begin{equation*}
        \alpha - X_0 \big(\phi(T- .) + \varrho\big) - \gamma u - (\mathbf{H} + \mathbf{G}) \circ (\mathbf{I} - \mathbf{R})(u)- \Big(\big(\mathbf{H} + \mathbf{G}\big) \circ \big(\mathbf{I} - D\mathbf{R}(u)\big)\Big)^* u - (\mathbf{H}_{\phi, \varrho} + \mathbf{H^*}_{\phi, \varrho} )u = 0.
    \end{equation*}
    % {\nathan where $\mathbf{H}_{\phi,\varrho}$ is the integral operator associated to the kernel
    % \begin{equation*}
    %     H_{\phi,\varrho}(t,s) := (\phi(T-t)+\varrho) \mathbbm{1}_{\{ t > s \}}, \quad s,t \in [0,T].
    % \end{equation*}

    In other words, if $u$ is a local extrema of $\calJ$, then it must satisfy the following nonlinear stochastic Fredholm system equation
    \begin{equation} \label{eq:foc_nonlinear_fredholm_resistance}
    \begin{cases}
        \gamma u + (\mathbf{H} + \mathbf{G} ) u + (\mathbf{H}_{\phi,\varrho}
    + \mathbf{H}_{\phi,\varrho}^{*})u - \mathbf{A}(u) = \alpha - X_0\big(\phi(T-.) + \varrho\big),\\
        % r^{u} - \calU ( \mathbf{G} ( u - r^{u} ) ) = 0,
        r^{u} = \mathbf{R}(u),
    \end{cases}
    \end{equation}
    where we introduce the following nonlinear operator $\mathbf{A}: \mathcal{L}^{2} \to \mathcal{L}^{2}$ such that
    \begin{equation} \label{eq:A_nonlinear_operator}
        \mathbf{A}(u) := (\mathbf{H} + \mathbf{G} ) \circ \mathbf{R}(u) - \Big( \big( \mathbf{I} + ( \mathbf{M}^u \circ \mathbf{G} )^{*} \big)^{-1} \circ \big(\mathbf{H} + \mathbf{G} \big)^{*} \Big)u, \quad u \in \mathcal{L}^{2},
    \end{equation}
    with $\mathbf{M}^u$ the multiplication operator defined by
    \begin{equation*} 
        \mathbf{M}^u v_{t} := \calU' \Big(\mathbf{G}\big(u-\mathbf{R}(u)\big)_{t}\Big) v_{t}, \quad t \in [0,T], \quad v,u \in \mathcal{L}^{2}.
    \end{equation*}
\end{theorem}

We now establish the existence of an optimal control.

\begin{theorem}
\label{thm:existence}

Suppose that Assumptions~\ref{assumption:G},~\ref{assumption:H} and~\ref{assumption:U} hold and suppose in addition that $\Omega$ is countable or finite and that $\kappa_{\infty} + \gamma > 0$, then there exists a global maximizer $\widehat{u} \in \mathcal{L}^2$ of the functional $\mathcal{J}$.
\end{theorem}

\section{Numerical scheme and application to optimal round-trips}\label{sec:numerical_applications}

In this section, we introduce an iterative numerical scheme for constructing 
\(M \in \mathbb{N}^{*}\) sample trajectories of a discrete-time approximation of the critical points \(\hat{u}, \hat{r}^{u} \in \mathcal{L}^{2}\) satisfying the First-Order Condition (FOC)~\eqref{eq:foc_nonlinear_fredholm_resistance}. We establish a theoretical convergence result for this scheme and illustrate its output through an example of optimal round-trip strategies in the presence of stochastic ``buy'' signals with various decay rates, as well as through a qualitative analysis of the influence of the resistance function’s convexity on optimal trading. We defer to Appendix~\ref{s:complement_numerical_results} the proof of the convergence result, together with additional illustrations of the qualitative effects of impact decay and permanent impact intensity on the resulting strategies.\\

In what follows, we fix a uniform time-grid $\mathbb{T}_{N} := \{ \frac{iT}{N}, \; i \in \{0, \ldots, N\} \}$, with $N \in \mathbb{N}^{*}$, and we use the notation $\mathbf{A}(u) = \mathbf{A}(u, r^{u})$ in order to explicitly refer to the dependence of $\mathbf{A}$ from~\eqref{eq:A_nonlinear_operator} on both $u$ and $r^{u}$, which will be constructed separately in practice.

\subsection{Approximation of operators} \label{ss:approx_operators}

\paragraph{Nyström approximation of Volterra integral operators.} Given a deterministic Volterra kernel $G : \mathbb{R} \to \mathbb{R}$, the integral operator $\mathbf{G} : \mathcal{L}^{2} \to \mathcal{L}^{2}$ and its dual operator $\mathbf{G}^{*}$ are numerically approximated on $\mathbb{T}_{N}$ by left-rectangles, while integrating the kernels in the same spirit as~\cite{nystrom1930praktische, abi2022optimal, abi2024trading, abi2025fredholm} such that
\begin{equation*}
    ( \mathbf{G}u )_{t_{i}} \approx \sum_{j=0}^{i-1} L_G(i,j) u_{t_{j}}, \quad ( \mathbf{G}^{*}u )_{t_{i}} \approx \sum_{j=0}^{i-1} M_G(i,j) \E_{t_{i}} u_{t_{j}}, \quad t_{i} \in \mathbb{T}_{N}, \quad u \in \mathcal{L}^{2},
\end{equation*}
where we define
\begin{equation*}
    L_{G} := \bigg( \int_{t_{j}}^{t_{j+1}} G(t_{i},s) \dd s \mathds{1}_{\{j \leq i\}} \bigg)_{i,j \in \{0, \ldots, N-1\}}, \quad M_{G} := \bigg( \int_{t_{j}}^{t_{j+1}} G(s,t_{i}) \dd s \mathds{1}_{\{j \geq i\}} \bigg)_{i,j \in \{0, \ldots, N-1\}}.
\end{equation*}

\paragraph{Numerical approximation of the nonlinear operator $\mathbf{A}$ from~\eqref{eq:A_nonlinear_operator}.}
Given $u \in \mathcal{L}^{2}$, let $r^{u} \in \mathcal{L}^{2}$ be the unique solution to the fixed-point equation in~\eqref{eq:foc_nonlinear_fredholm_resistance}. Then, the first term of $\mathbf{A}(u)$ in~\eqref{eq:A_nonlinear_operator}, i.e.,
\begin{equation*}
    (\mathbf{H} + \mathbf{G} ) r^{u}
\end{equation*}
is linear in $r^{u}$ and can be readily approximated via the Nyström technique, while it remains to estimate the nonlinear term given by
\begin{equation} \label{eq:non_linear_adjoint_approx}
    \Big( \big( \mathbf{I} + ( \mathbf{M}^u \circ \mathbf{G} )^{*} \big)^{-1} \circ (\mathbf{H} + \mathbf{G} )^{*} \Big)(u).
\end{equation}
To achieve this, we set
\begin{align} \notag
    & f := \big( \mathbf{I} + ( \mathbf{M}^u \circ \mathbf{G} )^{*} \big)^{-1} \circ (\mathbf{H} + \mathbf{G} )^{*} u  \\ \notag
    \iff & f + ( \mathbf{M}^u \circ \mathbf{G} )^{*} f = (\mathbf{H} + \mathbf{G} )^{*} u \\
    \label{eq:backward_fredholm}
    \iff & f_{t} + \int_{t}^{T} G(s,t) \E_{t} \Big[ \calU' \Big( \big(\mathbf{G} ( u - r^{u} ) \big)_{s} \Big) f_{s} \Big] \dd s = \int_{t}^{T} \big(H(s,t) + G(s,t)\big) \E_{t} u_{s} \dd s, \quad t \in [0,T]
\end{align}
where the explicit expressions of the dual operators are obtained by stochastic Fubini and the tower property, and are well-defined since $\mathbf{M}^u \circ \mathbf{G}$ and $(\mathbf{H} + \mathbf{G} )$ are both linear and bounded (since $\calU'$ is bounded), see \cite[Chapter 6, Section 2]{dunford1988linear}. Given $u \in \mathcal{L}^{2}$ and observing that $f_{T} = 0$, we can solve~\eqref{eq:backward_fredholm} by a backward iterative scheme on the subdivision $\mathbb{T}_{N}$ as follows:
\begin{itemize}
    \item $f_{t_{N}} = 0$,
    \item for $p \in \{ N-1, \cdots, 0 \}$, since $f_{t_{p}}$ and $u_{t_{p}}$ are $\mathcal{F}_{t_{p}}$-measurable, then a left-rectangle approximation of the integrals in~\eqref{eq:backward_fredholm} yields
    \begin{align} \notag
        f_{t_{p}} = \big(1 + M_{G}(p,p) w_{t_{p}}^{u} \big)^{-1} & \bigg( M_{G+H}(p,p) u_{t_{p}} + \sum_{k = p+1}^{N-1} M_{G+H}(p,k) \E_{t_{p}} u_{t_{k}} . \\ \label{eq:numerical_solution_f}
        & \quad . - \sum_{k = p+1}^{N-1} M_{G}(p,k) \E_{t_{p}} [ w_{t_{k}}^{u} f_{t_{k}}] \bigg),
    \end{align}
    where we defined
    \begin{equation*}
        w^{u} := \calU' \big(\mathbf{G} ( u - r^{u} ) \big),
    \end{equation*}
    and the conditional expectations $( \E_{t_{p}} [ w_{t_{k}}^{u} f_{t_{k}}] )_{N-1 \geq k > p \geq 0}, \; (\E_{t_{p}} u_{t_{k}})_{N-1 \geq k > p \geq 0}$ can be estimated by least-squares Monte Carlo similarly as in~\cite[Section 3.3]{abi2024trading} or~\cite[Section 3.1]{abi2025fredholm}.
\end{itemize}

\begin{rem}[Sanity check of the approximation of $\mathbf{A}$]
    The supremum $L^{2}$ error $E^{bf}$ associated to the backward Fredholm equation~\eqref{eq:backward_fredholm} is given by 
\begin{equation*}
    E^{bf} ( f ) := \sup_{\omega \in \Omega} \bigg(\int_0^T \big| f_{t} + \big( ( \mathbf{M}^u \circ \mathbf{G} )^{*} f \big)_{t} - \big( (\mathbf{H} + \mathbf{G} )^{*} u \big)_{t} \big|^2 \dd t\bigg] \bigg)(\omega),
\end{equation*}
can be estimated numerically on a uniform time grid with step $\Delta$ as
\begin{equation} \label{eq:max_error_backward_scheme}
    E_{N,M}^{bf} ( f ) := \sup_{m \in \{1, \ldots, M\}} E_{N,M}^{bf} ( f ) (\omega_{m})
\end{equation}
where
\begin{align*}
    E_{N,M}^{bf} ( f ) (\omega_{m}) & := \Delta \sum_{p=0}^{N-1} \bigg| f_{t_{p}}(\omega_{m}) + M_{G}(p,p) w_{t_{p}}^{u}(\omega_{m}) f_{t_{p}}(\omega_{m}) \\
    & + \sum_{k = p+1}^{N-1} M_{G}(p,k) \E_{t_{p}} [ w_{t_{k}}^{u} f_{t_{k}}](\omega_{m}) 
    - M_{G+H}(p,p) u_{t_{p}}(\omega_{m}) \\
    & - \sum_{k = p+1}^{N-1} M_{G+H}(p,k) \E_{t_{p}} u_{t_{k}}(\omega_{m}) \bigg|^2.
\end{align*}
In practice $E_{N,M}^{bf} ( f )$ from~\eqref{eq:max_error_backward_scheme} is way below machine precision, see the right plot in Figure~\ref{F:convergence_trading_buy_signals_various_decays}.
\end{rem}

\subsection{Iterative numerical scheme} 

\paragraph{Criterion of convergence.} For any given $u, r^{u} \in \mathcal{L}^{2}$, we define the joint supremum $L^{2}$ errors of the FOC~\eqref{eq:foc_nonlinear_fredholm_resistance} by
\begin{align*}
    E^{1} ( u, r^{u} )  &:= \sup_{\omega \in \Omega} \bigg( \int_{0}^{T} \Big|\gamma u_{t} + \big( (\mathbf{H} + \mathbf{G} + \mathbf{H}_{\phi,\varrho}
    + \mathbf{H}_{\phi,\varrho}^{*} ) u \big)_{t} - \big(\mathbf{A}(u, r^{u}) \big)_{t} \\
    & \qquad -\Big( \alpha_{t} - X_{0} \big( \phi(T-t) + \varrho \big) \Big) \Big|^{2} \dd t \bigg)(\omega), \\
    E^{2} ( u, r^{u} ) & := \sup_{\omega \in \Omega} \bigg(\int_{0}^{T} \Big|r_{t}^{u} - \calU \Big( \big(\mathbf{G} ( u - r^{u} ) \big)_{t} \Big) \Big|^{2} \dd t \bigg)(\omega),
\end{align*}
which, over a $N$-step time grid and $M$ sample trajectories, are approximated numerically by
\begin{equation} \label{eq:max_error_foc}
    E_{N,M}^{1} ( u^{N}, r^{N} ) := \sup_{m \in \{1, \ldots, M\}} E_{N,M}^{1} ( u^{N}, r^{N} ) (\omega_{m}),
\end{equation}
where
\begin{align*}
    E_{N,M}^{1} ( u^{N}, r^{N} ) (\omega_{m}) := & \Delta \sum_{i=0}^{N-1} \bigg| \gamma u_{t_{i}}^{N}(\omega_{m}) + \sum_{j=0}^{i-1} L_{G+H+H_{\phi, \varrho}}(i,j) u_{t_{j}}^{N}(\omega_{m}) + M_{H_{\phi, \varrho}}(i,i) u_{t_{i}}^{N} \\
    & \qquad \quad + \sum_{j = i+1}^{N-1} M_{H_{\phi, \varrho}}(i,j) \E_{t_{i}} u_{t_{j}}^{N} (\omega_{m}) - \sum_{j=0}^{i-1} L_{G+H}(i,j) r_{t_{j}}^{N}(\omega_{m}) + f_{t_{i}}^{N}(\omega_{m}) \\
    & \qquad \quad - \alpha_{t_{i}}(\omega_{m}) + X_{0} ( \phi(T-t_{i}) + \varrho ) \bigg|^2,
\end{align*}
with $f_{t_{i}}^{N}$ is given by~\eqref{eq:numerical_solution_f}, and
\begin{equation} \label{eq:max_error_resistance}
    E_{N,M}^{2} ( u^{N}, r^{N} ) := \sup_{m \in \{1, \ldots, M\}} E_{N,M}^{2} ( u^{N}, r^{N} ) (\omega_{m}),
\end{equation}
where
\begin{equation*}
    E_{N,M}^{2} ( u^{N}, r^{N} ) (\omega_{m}) := \Delta \sum_{i=0}^{N-1} \bigg| r_{t_{i}}^{N}(\omega_{m}) - \calU \bigg( \sum_{j=0}^{i-1} L_{G}(i,j) \big( u_{t_{j}}^{N}(\omega_{m}) - r_{t_{j}}^{N}(\omega_{m}) \big) \bigg) \bigg|^2.
\end{equation*}
In what follows, we drop the upper-script $N$ from $u^{N}, r^{N}$ for ease of reading. We then aim to construct a sequence of such discrete-time approximations of $( u, r^{u})$ which we denote by $( u^{[n]}, \; r^{[n]} )_{n \geq 0} = ( u^{[n]}(N,M), \; r^{[n]}(N,M) )_{n \geq 0}$ such that 
\begin{equation*}
    E_{N,M}^{1,2} \big( u^{[n]}(N,M), \; r^{[n]}(N,M) \big) \underset{n \to \infty}{\longrightarrow} 0.
\end{equation*}

\begin{comment}
\fbox{%
  \begin{minipage}{0.9\textwidth}
\textbf{IDEAL RESULT (Is this true?):} We would like to show that for any $N,M \in \mathbb{N}$, the sequence $( u^{[n]}(N,M), r^{[n]}(N,M) )_{n \geq 0}$ converges to a unique limit $( u^{[\infty]}(N,M), r^{[\infty]}(N,M) )_{n \geq 0}$ such that
\begin{equation*}
    ( u^{[\infty]}(N,M), \; r^{[\infty]}(N,M) ) \underset{n, N, M \to \infty}{\longrightarrow} ( \hat{u}, \hat{r}^{\hat{u}} ),
\end{equation*}
with $( \hat{u}, \hat{r}^{\hat{u}} )$ a solution to~\eqref{eq:foc_nonlinear_fredholm_resistance} (WE NEED AN EXISTENCE RESULT :P).
\end{minipage}%
}
\end{comment}

\paragraph{Scheme definition.}

Fix $\epsilon > 0$.
\begin{itemize}
    \item Initialize $( u^{[0]}, r^{[0]} ) \equiv 0$.
    \item While $E_{N,M}^{1} ( u^{[n]}, r^{[n]} ) > \epsilon$, do:
    \begin{itemize}
        \item Update $u^{[n]}$ by solving~\eqref{eq:foc_nonlinear_fredholm_resistance} while computing the nonlinear operator $\mathbf{A}$ using $u^{[n-1]}, \; r^{[n-1]}$:
        \begin{equation} \label{eq:approximate_FOC_scheme}
            \gamma u^{[n]} + (\mathbf{H} + \mathbf{G} + \mathbf{H_{\phi, \varrho}} + \mathbf{H_{\phi, \varrho}^{*}} ) u^{[n]} = \alpha - X_{0} \big( \phi(T-t_{i}) + \varrho \big) + \mathbf{A} (u^{[n-1]}, r^{[n-1]} ),
        \end{equation}
        which can be numerically solved as detailed in~\cite[Section 3.2]{abi2024trading}.
    
        \item Update the market resistance $r^{[n]} = r^{[n,\infty]}$ as the limit of the sequence $( r^{[n,p]} )_{p \geq 0}$ solving the fixed-point equation from~\eqref{eq:foc_nonlinear_fredholm_resistance} via Picard iterations over $p \geq 0$ until $E_{N,M}^{2} ( u^{[n]}, r^{[n,p]} )$ is smaller than $\epsilon$:
        \begin{equation} \label{eq:update_resistance_scheme}
        \begin{cases}
            r^{[n,p+1]} = \calU \big( \mathbf{G} ( u^{[n]} - r^{[n,p]} ) \big), \\
            r^{[n,0]} := r^{[n-1,0]}.
        \end{cases}
        \end{equation}

        \item Increment $n$ to $n+1$.
    \end{itemize}
\end{itemize}

\paragraph{Theoretical convergence of the iterative scheme.} 
In the following, we state a result on the theoretical rate of convergence of the iterative scheme if $T$ is sufficiently small or, equivalently, if $\gamma$ is sufficiently large, in the same spirit as~\cite[Proposition 2.14]{abi2025fredholm}.
\begin{theorem}[Exponential convergence of the iterative scheme] \label{T:scheme_convergence}
    Assume the resistance function $\calU$ is Lipschitz-continuous with constant $L > 0$, with bounded derivative $\calU' < C$, for some positive constant $C$. Fix two admissible kernels $H, G$, as well as the trading horizon $T > 0$ and the slippage costs parameter $\gamma > 0$ such that
    \begin{align} \label{eq:condition_1_for_cvge}
        1 & > \sqrt{T C_{G}}\max(L,C), \\ \label{eq:condition_2_for_cvge}
        \gamma & > \tilde{C} := \sqrt{T C_{H+G}} \bigg( L\frac{\sqrt{T C_{G}}}{1-L\sqrt{T C_{G}}} + \frac{1}{1-C\sqrt{T C_{G}}} \bigg),
    \end{align}
    where $C_{H+G}$ and $C_{G}$ are respectively the constants associated to the kernels $H+G$ and $G$ defined by~\eqref{eq:constant_norm_of_G}. Assume the existence of $( \hat{u}, \hat{r} )$ satisfying the first-order-condition~\eqref{eq:foc_nonlinear_fredholm_resistance}, and denote by $( u^{[n]}, r^{[n]} )_{n \geq 0}$ the sequence of controls obtained by the iterative scheme~\eqref{eq:approximate_FOC_scheme}--\eqref{eq:update_resistance_scheme}, then
    $$\lim_{n\to\infty} u^{[n]} = \hat{u}\quad \text{in $\mathcal{L}^2$,}$$ 
    and the convergence rate is bounded by
    \begin{equation} \label{eq:convergence_rate}
      \norm{u^{[n]}-\hat{u}}  \le \bigg(\frac{\widetilde{C}}{\gamma}\bigg)^n\norm{\hat{u}}, \quad n \in \mathbb{N}.
    \end{equation}
    As a consequence, $( \hat{u}, \hat{r} )$ is the unique solution to the FOC~\eqref{eq:foc_nonlinear_fredholm_resistance}.
\end{theorem}

\begin{proof}
    See Appendix~\ref{ss:proof_convergence_scheme}.
\end{proof}

\begin{rem}[Stochastic gradient descent]
    Starting from $u^{[0]} = 0$, consider a stochastic gradient iterative scheme of the form
    \begin{equation*}
        u^{[n]} = u^{[n-1]} + \eta_{n} \nabla \mathcal{J}(u^{[n-1]}), \quad n \geq 1,
    \end{equation*}
    where $\left( \eta_{n} \right)$ is a sequence of learning steps and $\nabla \mathcal{J}$ is given by~\eqref{eq:gateaux_derivative}. Although appearing to be a promising candidate to construct critical points satisfying the first-order optimality condition~\eqref{eq:foc_nonlinear_fredholm_resistance} without having to invert any Fredholm equation, we observed in practice that such scheme converges slower than the iterative scheme~\eqref{eq:approximate_FOC_scheme}--\eqref{eq:update_resistance_scheme} and does not converge at all in the regime $\gamma \to 0$. 
\end{rem}

\begin{comment}
    \subsection{Stochastic gradient ascent}

Given $u \in \mathcal{L}^{2}$, recall the gradient is given by
\begin{equation*}
\nabla \mathcal{J}(u, r^{u}) := \alpha - X_{0} ( \phi(T- \cdot) + \varrho ) - \gamma u - ( \mathbf{H} + \mathbf{G} + \mathbf{H}_{\phi,\varrho}
    + \mathbf{H}_{\phi,\varrho}^{*} ) u + \mathbf{A}(u, r^{u}),
\end{equation*}
where the resistance $r^{u}$ solves the fixed-point equation~\eqref{eq:foc_nonlinear_fredholm_resistance} and $\mathbf{A}(u, r^{u})$ is defined by~\eqref{eq:A_nonlinear_operator}.

\paragraph{Scheme definition.} Fix $\epsilon > 0$.
\begin{itemize}
    \item Initialization: $u^{[0]} \equiv 0$.

    \item While $E_{N,M}^{1} ( u^{[n]}, r^{[n]} ) > \epsilon$:
    \begin{enumerate}
        \item Update the trading strategy:
        \begin{equation*}
            u^{[n]} = u^{[n-1]} + \eta_{n} \nabla \mathcal{J}(u^{[n-1]}, r^{[n-1]}),
        \end{equation*}
        where $\eta_{n} := \frac{\eta}{n^{\beta}}$ is the learning step.

        \item Update the market resistance $r^{[n]} = r^{[n,\infty]}$ as the limit of the sequence $( r^{[n,p]} )_{p \geq 0}$ solving its fixed-point equation via Picard iterations over $p \geq 0$ until $E_{N,M}^{2} ( u^{[n]}, r^{[n,p]} )$ is smaller than $\epsilon$:
        \begin{equation*}
        \begin{cases}
            r^{[n,p+1]} = \calU ( \mathbf{G} ( u^{[n]} - r^{[n,p]} ) ) \\
            r^{[n,0]} := r^{[n-1,0]}
        \end{cases}.
        \end{equation*}

        \item Increment $n$ to $n+1$.
    \end{enumerate}
\end{itemize}
\end{comment}

\subsection{Optimal round-trips in presence of ``buy'' signals} \label{ss:application_to_optimal_round_trips}

\paragraph{Impact model specification.} Given fixed constants $\delta > 0$ and $c \geq 1$, we introduce the resistance function $\calU_{\delta,c} : \mathbb{R} \to \mathbb{R}$ as
\begin{equation} \label{eq:resistance_specification}
    \calU_{\delta,c}(x) := sign(x) |x|^{c} \mathds{1}_{\{|x| \leq \delta \}} + \big( c\delta^{c-1}x - sign(x)\delta^{c} ( c - 1 ) \big) \mathds{1}_{\{|x| > \delta \}}, \quad x \in \mathbb{R},
\end{equation}
and its derivative is
\begin{equation} \label{eq:derivative_resistance_specification}
    \calU_{\delta,c}'(x) := c |x|^{c-1} \mathds{1}_{\{|x| \leq \delta \}} + c\delta^{c-1} \mathds{1}_{\{|x| > \delta \}}, \quad x \in \mathbb{R}.
\end{equation}
The parameter \(c \ge 1\) governs the degree of convexity, encompassing in particular the linear (\(c=1\)) and quadratic (\(c=2\)) resistance functions. We also note that, for \(c=2\), the specification~\eqref{eq:resistance_specification} coincides with a modified Huber loss~\cite{huber1992robust}, in which the sign is flipped for negative arguments so that market impact indeed moves the unaffected price unfavorably when executing large sell orders. A finite value for $\delta < \infty$ yields an asymptotic linear growth of the resistance function as in Assumption~\ref{assumption:U}, and guarantees Lemma~\ref{lem:frechet:R} and Theorem~\ref{T:scheme_convergence} are applicable. Although $\delta$ might be a desirable tunable regularizing model parameter, we do not need to explicitly specify it in practice when considering a finite number of bounded sample signal trajectories (as in the following numerical examples), and therefore $\delta$ may remain implicit such that only the convex term $sign(\cdot)|\cdot|^{c} \mathds{1}_{\{|\cdot| \leq \delta \}}$ in~\eqref{eq:resistance_specification} is effective. 

\begin{figure}[H]
\centering
\includegraphics[width=6in]{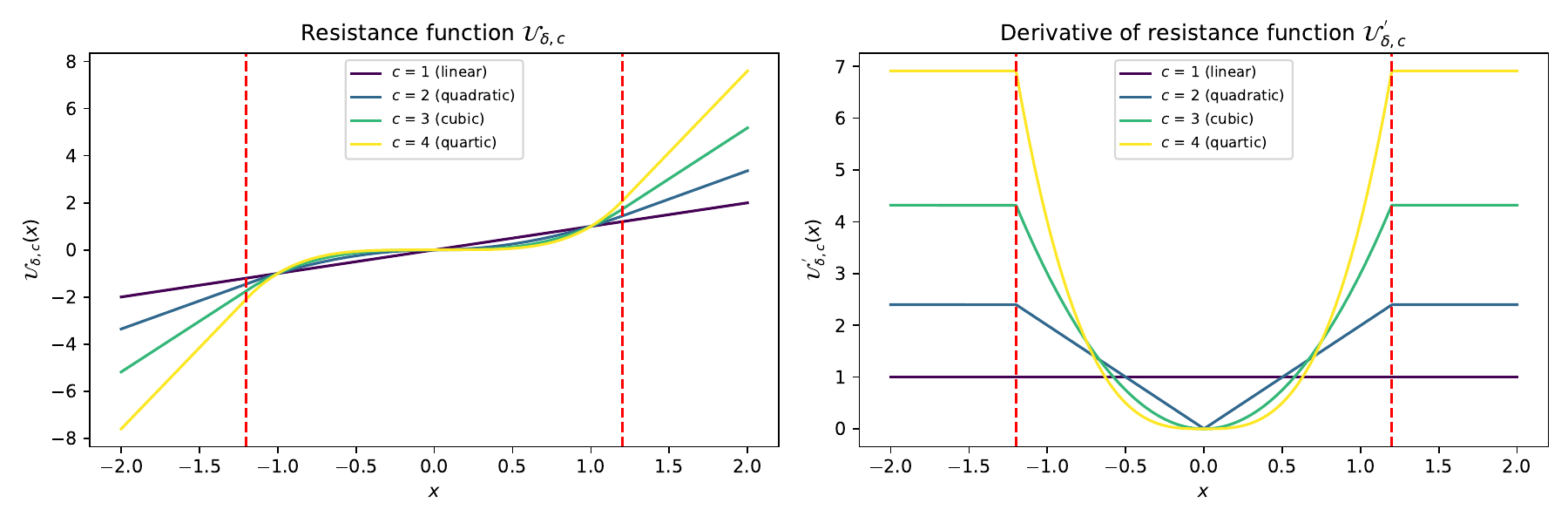}
\caption{Resistance function specification~\eqref{eq:resistance_specification} and its derivative~\eqref{eq:derivative_resistance_specification} for various values of $c$, while fixing $\delta = 1.2$.}
\label{F:resistance_function_and_its_derivative}
\end{figure}

Unless stated otherwise, the parameters of the resistance function are set as follows:
\begin{itemize}
    \item $c = 2$ for the convexity of the resistance function; this corresponds to what we refer to as the \textit{Quadratic Market Resistance} (QMR) model, asymptotically consistent with the square-root law in participation rate as shown in~\cite[Theorem 9]{durin2023two};
    \item $\delta < \infty$ is left implicit, since we consider bounded sample signal trajectories.
\end{itemize}

As motivated in Section~\ref{sec:market_impact_model}, we specify a power-law decay kernel $G_{\lambda, \nu}$ to capture the transient nature of market impact. Therefore, unless stated otherwise, the parameters of the price impact model in~\eqref{eq:price_impact_model} are set as follows:
\begin{itemize}
\item $\gamma = 0.2$ for the intensity of slippage costs;
\item $\lambda = 0.467$ and $\nu = 0.614$ for the transient impact decay (borrowed from~\cite{abi2025fredholm}, Figure~7);
\item $\kappa_{\infty} = 1$ for the permanent impact intensity.
\end{itemize}

\paragraph{Soft constraints specification.} The trader has the possibility to penalize his running inventory as well as the final inventory level via the hyper-parameters $\phi$ and $\varrho$ respectively introduced in~\eqref{eq:def:J}. Unless stated otherwise, we set:
\begin{itemize}
    \item $\phi=0$, i.e., no penalty on the running inventory, since it is obvious that an increasing value of $\phi$ shrinks the optimal inventory toward $0$;
    
    \item $\varrho = 5e2 \gg 1$ to enforce a zero final inventory.
\end{itemize}

\paragraph{Signals and least-squares Monte Carlo specification.} We specify an Ornstein-Uhlenbeck (OU) drift-like price signal $\mu$ and set $P := \int \mu_{t} \dd t$ where $P$ is given in the unaffected price process~\eqref{eq:unaffected_price_process}. We have
\begin{equation} \label{eq:drift_signal_specification}
\dd \mu_t = (\eta - \kappa \mu_t) \dd t + \sigma \dd W_t, \quad \mu_{0} \in \mathbb{R}.
\end{equation}
The alpha-signal $\alpha$ from~\eqref{eq:alpha_drift} is then explicitly given by
\begin{equation} \label{eq:signal_specification}
\alpha_{t} := \mathbb{E}_{t} \bigg[ \int_{t}^{T} \mu_r \dd r \bigg] = \bigg( \mu_{t} - \frac{\eta}{\kappa} \bigg) \frac{ 1 - e^{- \kappa (T-t)}}{\kappa} + \frac{\eta}{\kappa} (T-t), \quad t \in [0,T].
\end{equation}
Unless stated otherwise, we set:
\begin{itemize}
    \item $\eta=10$, for the long-term mean;
    
    \item $\kappa = 1$ for the mean-reversion rate;

    \item Either $\sigma = 1$ or $\sigma = 0$ for the signal noise level, see Remark~\ref{R:setting_vol_to_zero_for_clarity_and_stability};

    \item $\mu_{0} = 1$ for the initial state value of the OU drift-signal.
\end{itemize}
In the case of stochastic OU drift-signals, i.e., $\sigma = 1$, we need to apply a least-squares Monte Carlo to approximate the nonlinear stochastic operator~\eqref{eq:non_linear_adjoint_approx} via~\eqref{eq:numerical_solution_f} as explained in Section~\ref{ss:approx_operators}. In what follows, we apply Ridge regressions on $2000$ signal sample trajectories with penalty $1e-5$ to a basis expansion of all Laguerre polynomials up to degree two of the following family of features
\begin{equation*}
    \bigg( \alpha, \int_{0}^{\cdot} \alpha_{s} \dd s, \int_{0}^{\cdot} e^{- \kappa (\cdot - s)} \alpha_{s} \dd s \bigg),
\end{equation*}
and refer the interested reader to~\cite[Section 3.3]{abi2024trading} and~\cite[Section 3.1 and Appendix A]{abi2025fredholm} for more details. Finally, the number of time steps is set to $N=100$.

\begin{rem} \label{R:setting_vol_to_zero_for_clarity_and_stability}
    Setting the signal volatility to zero yields a deterministic optimal trading problem, leading to the same qualitative conclusions as those obtained from the sample means of the optimal quantities, while improving visual clarity in practice and avoiding the error propagation inherent to the least-squares Monte Carlo estimation procedure.
\end{rem}

\paragraph{Optimal round-trips with the QMR model for various signal decays.}

We specify three stochastic nonnegative ``buy'' signals of the form~\eqref{eq:signal_specification} and vary the signal's mean-reversion $\kappa \in \{ 0.1, 1, 10\}$ from~\eqref{eq:drift_signal_specification}, as illustrated in the top-left plot in Figure~\ref{F:trading_buy_signals_various_decays}.\\

On the one hand, Figure~\ref{F:trading_buy_signals_various_decays} displays the sample averages and the associated $95\%$ normal confidence intervals of the main optimal quantities of interest associated to the optimal round-trip strategies: optimal trading rates, inventories and resistance rates, as well as the resulting price distortions and running trading costs. We also display $5$ sample trajectories out of the $2000$ ones used for the least-squares Monte Carlo. The slower the signal decay, the more aggressive the trading strategy becomes, and the greater the resulting market resistance, price distortion, and running trading costs. Also, note that the resulting trading strategies do not feature any model round-trip model arbitrage as all the sample running trading costs trajectories remain non-negative throughout the trading horizon.\\

On the other hand, Figure~\ref{F:convergence_trading_buy_signals_various_decays} illustrates the exponential convergence of the iterative scheme through iterations: the faster the signal decays, the faster the scheme converges. Note that the convergence rate depends on all the impact model's parameters as suggested by inequality~\eqref{eq:condition_2_for_cvge} from Theorem~\ref{T:scheme_convergence}. In particular, we observe that the smaller $\gamma$, the slower the convergence -- until numerical instability appears -- which is consistent with the regularizing effect of slippage costs already described in~\cite[Section 3.5]{abi2025fredholm}.

\begin{figure}[H]
\centering
\includegraphics[width=6in]{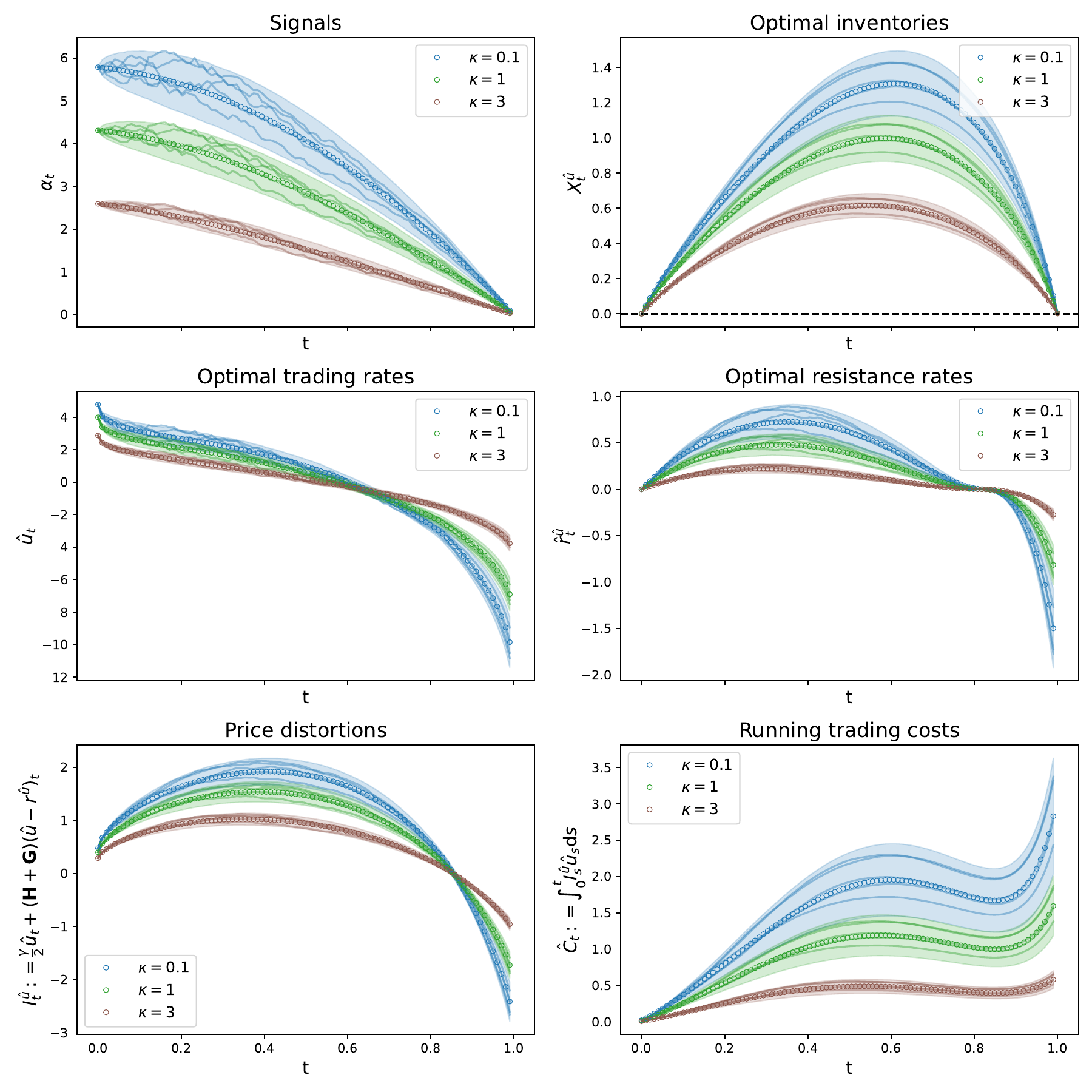}
\caption{Optimal round-trips for three stochastic ``buy'' signals with different signal decays $\kappa$ from~\eqref{eq:drift_signal_specification}. For each quantity, the shaded regions represent their normal $95\%$ confidence intervals estimated from the optimal $2000$ sample trajectories, the empty dot markers denote the corresponding sample means, and we also display $5$ sample trajectories.}
\label{F:trading_buy_signals_various_decays}
\end{figure}

\begin{figure}[H]
\centering
\includegraphics[width=6in]{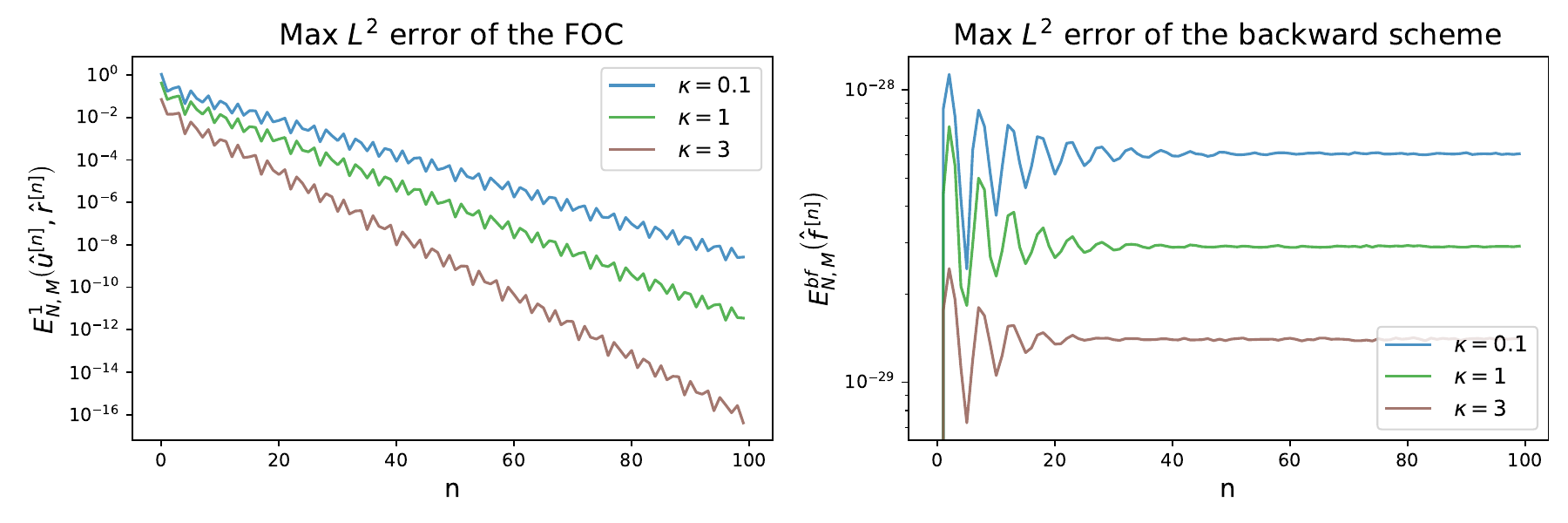}
\caption{Convergence of the numerical scheme~\eqref{eq:approximate_FOC_scheme}--\eqref{eq:update_resistance_scheme} for three stochastic ``buy'' signals with different signal decays $\kappa$ from~\eqref{eq:drift_signal_specification}: the left-hand plot displays the error $E_{N,M}^{1}$ of the FOC defined in~\eqref{eq:max_error_foc} while the right-hand plot shows the numerical error $E_{N,M}^{bf}$ of the backward scheme~\eqref{eq:numerical_solution_f} defined in~\eqref{eq:max_error_backward_scheme} as a function of the iteration step of the scheme. Note that the numerical error $E_{N,M}^{2}$, defined in from~\eqref{eq:max_error_resistance}, and due to Picard iterations when computing the resistance function is set to be lower than $1e-16$ at each iteration.}
\label{F:convergence_trading_buy_signals_various_decays}
\end{figure}

\paragraph{Optimal round-trips for different convexity parameters.}

For the remaining numerical results, including those reported in Appendix~\ref{ss:trading_sensivity_to_model_parameters}, we set the signal noise level to zero, that is, \(\sigma = 0\) in~\eqref{eq:drift_signal_specification}, and run the iterative scheme on a single deterministic signal trajectory. All other parameters are kept fixed as in Section~\ref{ss:application_to_optimal_round_trips}, unless stated otherwise; see Remark~\ref{R:setting_vol_to_zero_for_clarity_and_stability}.\\

Furthermore, in all subsequent numerical experiments, the convergence errors \(E_{N,M}^{1}\) in~\eqref{eq:max_error_foc}, \(E_{N,M}^{2}\) in~\eqref{eq:max_error_resistance}, and \(E_{N,M}^{\mathrm{bf}}\) in~\eqref{eq:max_error_backward_scheme} are found to be below \(10^{-11}\), \(10^{-16}\), and \(10^{-31}\), respectively, after \(100\) iterations of the numerical scheme~\eqref{eq:approximate_FOC_scheme}--\eqref{eq:update_resistance_scheme}.\\

Figure~\ref{F:trading_buy_signals_various_convexity} illustrates the qualitative effect of the convexity parameter $c \geq 1$ of the resistance function $\calU_{M,c}$ from~\eqref{eq:resistance_specification}: in this case, the higher the value of $c$, the more aggressive the optimal trading strategy, and the longer the optimal resistance rate lingers near zero when the trader changes his trading direction, which is consistent with the shapes of $\calU_{\delta, c}, \; c \in \{ 1, 2, 3, 4 \}$ depicted in Figure~\ref{F:resistance_function_and_its_derivative}. Therefore, the convexity captures how easily the market resistance reacts to the trader's strategy, and can be interpreted as a proxy of the ability of a sophisticated trader to detect the executed volume. Compared to the linear propagator model shown in red (i.e., the case \(\mathcal{U} \equiv 0\)), introducing a market resistance leads to less aggressive trading, which is consistent with the interpretation of market resistance as a proxy for sophisticated traders who benefit from the trader's metaorder’s impact and thereby erode part of the trader's alpha.\\

\begin{figure}[H]
\centering
\includegraphics[width=6in]{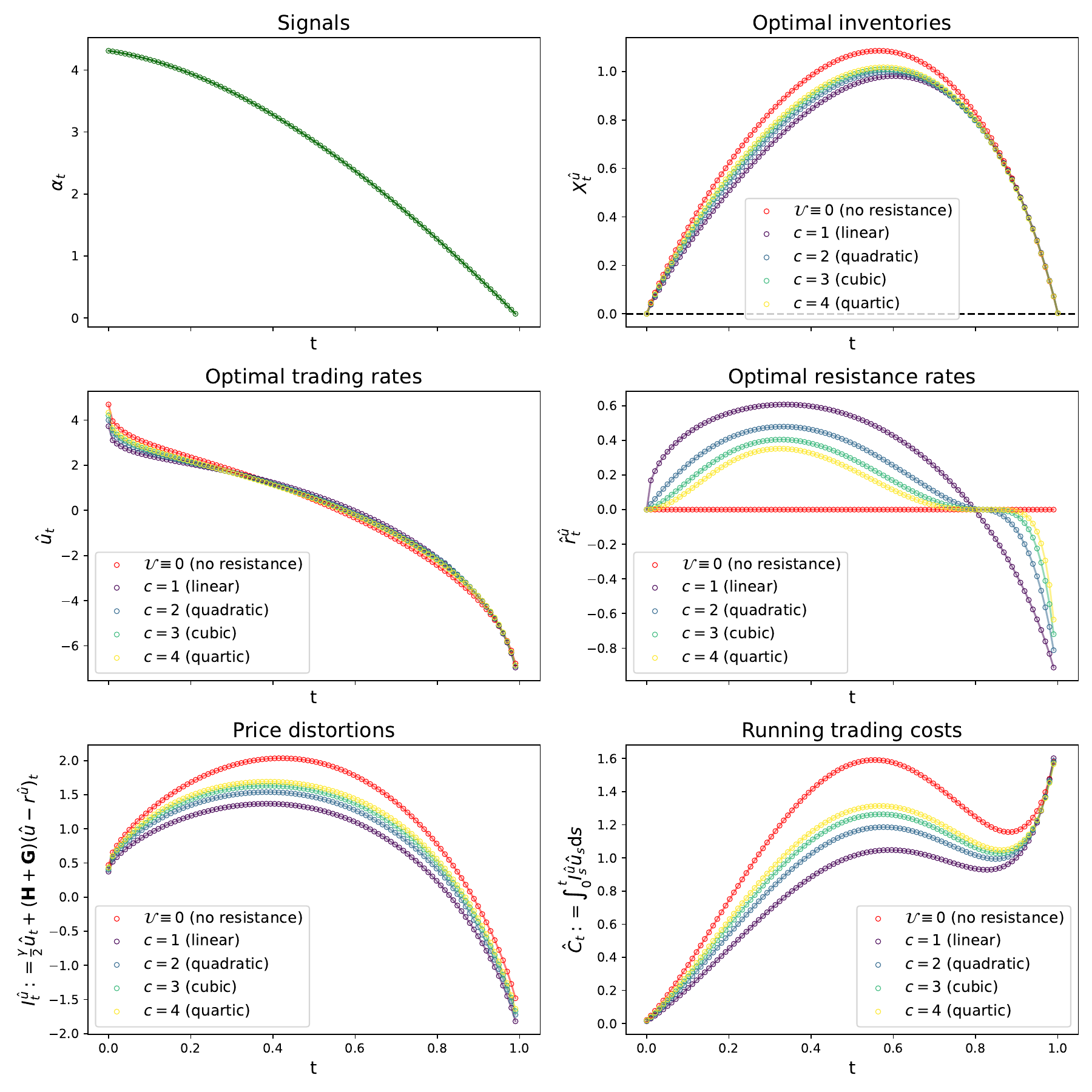}
\caption{Optimal round-trips in absence of resistance $(\calU \equiv 0)$ and with different convexity parameters $c \in \{ 1,2,3,4 \}$ used in $\calU_{\delta, c}$ from~\eqref{eq:resistance_specification}.}
\label{F:trading_buy_signals_various_convexity}
\end{figure}

\bibliographystyle{plainnat}
\bibliography{full_library.bib}

\appendix

\section{Complement to numerical results} \label{s:complement_numerical_results}

\subsection{Proof of Theorem~\ref{T:scheme_convergence}} \label{ss:proof_convergence_scheme}

Fix $n \in \mathbb{N}^{*}$. Subtracting the FOC~\eqref{eq:foc_nonlinear_fredholm_resistance} satisfied by $( \hat{u}, \hat{r} )$ to the scheme equations~\eqref{eq:approximate_FOC_scheme}--\eqref{eq:update_resistance_scheme} defining $( u^{[n]}, r^{[n]} )$ yields
\begin{equation} \label{eq:difference_system_foc_scheme}
    \begin{cases}
        \gamma ( \hat{u} - u^{[n]} ) + (\mathbf{H}+\mathbf{G}+\mathbf{H}_{\phi,\varrho}+\mathbf{H}_{\phi,\varrho}^{*}) ( \hat{u} - u^{[n]} ) = \mathbf{A}(\hat{u}, \hat{r}) - \mathbf{A}(u^{[n-1]}, r^{[n-1]}), \\
        \hat{r} - r^{[n-1]} = \calU \big( \mathbf{G} ( \hat{u} - \hat{r} ) \big) - \calU \big( \mathbf{G} ( u^{[n-1]} - r^{[n-1]} ) \big).
    \end{cases}
\end{equation}
On the one hand, take the inner product of the second equation from~\eqref{eq:difference_system_foc_scheme} against the process $\hat{r} - r^{[n-1]}$ such that
\begin{align*}
    \norm{\hat{r} - r^{[n-1]}}^{2} & = \langle \calU \big( \mathbf{G} ( \hat{u} - \hat{r} ) \big) - \calU \big( \mathbf{G} ( u^{[n-1]} - r^{[n-1]} ) \big), \hat{r} - r^{[n-1]} \rangle \\
    & \leq \norm{\calU \big( \mathbf{G} ( \hat{u} - \hat{r} ) \big) - \calU \big( \mathbf{G} ( u^{[n-1]} - r^{[n-1]} ) \big)} \norm{\hat{r} - r^{[n-1]}} \\
    & \leq L \norm{\mathbf{G} ( \hat{u} - u^{[n-1]} ) + \mathbf{G} ( r^{[n-1]} - \hat{r} )} \norm{\hat{r} - r^{[n-1]}} \\
    & \leq L \sqrt{T C_{G}} \norm{\hat{u} - u^{[n-1]}} \norm{\hat{r} - r^{[n-1]}} + L \sqrt{T C_{G}} \norm{\hat{r} - r^{[n-1]}}^{2}
\end{align*}
where the respective inequalities are derived by applying Cauchy-Schwartz, and by using the Lipschitz-continuity of $\calU$, the triangular inequality, and the upper bound~\eqref{eq:estimate_on_admissible_G_operator}. Therefore, using~\eqref{eq:condition_1_for_cvge}, we deduce
\begin{equation} \label{eq:bound_norm_resistance_foc_scheme}
    \norm{\hat{r} - r^{[n-1]}} \leq \frac{L \sqrt{T C_{G}}}{1-L \sqrt{T C_{G}}} \norm{\hat{u} - u^{[n-1]}}.
\end{equation}

On the other hand, take the inner product of the first equation from~\eqref{eq:difference_system_foc_scheme} with the process ($\hat{u} - u^{[n]})$ such that
\begin{align*}
    \gamma \norm{\hat{u} - u^{[n]}}^{2} & = - \langle (\mathbf{H}+\mathbf{G}+\mathbf{H}_{\phi,\varrho}+\mathbf{H}_{\phi,\varrho}^{*}) ( \hat{u} - u^{[n]} ), \hat{u} - u^{[n]} \rangle\\
    & \quad + \langle \mathbf{A}(\hat{u}, \hat{r}) - \mathbf{A}(u^{[n-1]}, r^{[n-1]}), \hat{u} - u^{[n]} \rangle \\
    & \leq \norm{\mathbf{A}(\hat{u}, \hat{r}) - \mathbf{A}(u^{[n-1]}, r^{[n-1]})} \norm{\hat{u} - u^{[n]}},
\end{align*}
where we used the positive semi-definite property of $(\mathbf{H}+\mathbf{G}+\mathbf{H}_{\phi,\varrho}+\mathbf{H}_{\phi,\varrho}^{*})$, as well as the Cauchy-Schwartz inequality. Therefore, we have
\begin{align} \notag
    \gamma \norm{\hat{u} - u^{[n]}} & \leq \norm{\mathbf{A}(\hat{u}, \hat{r}) - \mathbf{A}(u^{[n-1]}, r^{[n-1]})} \\ \notag
    & = \Big| \Big| (\mathbf{H}+\mathbf{G}) ( \hat{r} - r^{[n-1]} ) + \big( \mathbf{Id} + ( \mathbf{M}^{u^{[n-1]}} \circ \mathbf{G} )^{*} \big)^{-1} \big( (\mathbf{H}+\mathbf{G})^{*} u^{[n-1]} \big) \\
    & \quad - \big( \mathbf{Id} + ( \mathbf{M}^{\hat{u}} \circ \mathbf{G} )^{*} \big)^{-1} \big( (\mathbf{H}+\mathbf{G})^{*} \hat{u} \big) \Big| \Big| \\ \label{eq:last_ineq_bounding_error}
    & \leq \sqrt{T C_{H+G}} \norm{\hat{r} - r^{[n-1]}} + \norm{f^{[n-1]}-\hat{f}},
\end{align}
where we set
\begin{align*}
    f^{[n-1]} & := \big( \mathbf{Id} + (  \mathbf{M}^{u^{[n-1]}} \circ \mathbf{G} )^{*} \big)^{-1} \big( (\mathbf{H}+\mathbf{G})^{*} u^{[n-1]} \big), \\
    \hat{f} & := \big( \mathbf{Id} + ( \mathbf{M}^{\hat{u}} \circ \mathbf{G} )^{*} \big)^{-1} \big( (\mathbf{H}+\mathbf{G})^{*} \hat{u} \big).
\end{align*}
Note that we have equivalently
\begin{align} \label{eq:eq_f_n_minus_one}
    f^{[n-1]} + ( \mathbf{M}^{u^{[n-1]}} \circ \mathbf{G} )^{*} f^{[n-1]} & := (\mathbf{H}+\mathbf{G})^{*} u^{[n-1]}, \\ \label{eq:eq_f_hat}
    \hat{f} + ( \mathbf{M}^{\hat{u}} \circ \mathbf{G} )^{*} \hat{f} & := (\mathbf{H}+\mathbf{G})^{*} \hat{u}.
\end{align}
Therefore, subtracting~\eqref{eq:eq_f_hat} to~\eqref{eq:eq_f_n_minus_one}, and taking the inner product with $f^{[n-1]}-\hat{f}$ yields
\begin{align*}
    \norm{f^{[n-1]}-\hat{f}}^{2} & = \langle (\mathbf{H}+\mathbf{G})^{*} ( u^{[n-1]} - \hat{u} ), f^{[n-1]}-\hat{f} \rangle - \langle ( \mathbf{M}^{u^{[n-1]}} \circ \mathbf{G} )^{*} f^{[n-1]} - ( \mathbf{M}^{\hat{u}} \circ \mathbf{G} )^{*} \hat{f}, f^{[n-1]}-\hat{f} \rangle \\
    & \leq \norm{(\mathbf{H}+\mathbf{G})^{*} ( u^{[n-1]} - \hat{u} )} \norm{f^{[n-1]}-\hat{f}} + C \norm{\mathbf{G} ( \hat{f} - f^{[n-1]} )} \norm{f^{[n-1]}-\hat{f}} \\
    & \leq \sqrt{T C_{H+G}} \norm{u^{[n-1]} - \hat{u}} + C \sqrt{T C_{G}} \norm{\hat{f} - f^{[n-1]}},
\end{align*}
where the inequalities are respectively obtained by applying Cauchy-Schwartz, and using the boundedness of $\calU'$ and the bounds~\eqref{eq:estimate_on_admissible_G_operator}--\eqref{eq:estimate_on_adjoint}. Thus, using~\eqref{eq:condition_1_for_cvge}, we obtain
\begin{equation} \label{eq:bound_f_n_minus_one_minus_hat_f}
    \norm{f^{[n-1]}-\hat{f}} \leq \frac{\sqrt{T C_{H+G}}}{1-C\sqrt{T C_{G}}} \norm{u^{[n-1]} - \hat{u}}.
\end{equation}
Finally, injecting~\eqref{eq:bound_f_n_minus_one_minus_hat_f} and~\eqref{eq:bound_norm_resistance_foc_scheme} into~\eqref{eq:last_ineq_bounding_error} readily leads to
\begin{equation*}
    \norm{\hat{u} - u^{[n]}} \leq \frac{\tilde{C}}{\gamma} \norm{u^{[n-1]} - \hat{u}},
\end{equation*}
where $\tilde{C}$ is given by~\eqref{eq:condition_2_for_cvge} such that $(\norm{\hat{u} - u^{[n]}} )_{n}$ is sub-geometric with common ratio $\frac{\tilde{C}}{\gamma}$ and therefore converges to zero at exponential rate given by~\eqref{eq:convergence_rate}.

\subsection{Qualitative effects of impact decay and permanent impact} \label{ss:trading_sensivity_to_model_parameters}

Figure~\ref{F:trading_buy_signals_various_impact_decay} illustrates the qualitative effect of the impact decay parameter $\nu \in ( \frac{1}{2}, 1 )$ of the power-law kernel $G_{\lambda, \nu}$ from~\eqref{eq:decay_specification}: the larger $\nu$ is, the faster the decay becomes, and consequently the more aggressive the resulting optimal trades are.\\

Finally, Figure~\ref{F:trading_buy_signals_various_permanent_impact} illustrates the qualitative effect of the $\kappa_\infty \in \{ 0.5, 1, 1.5 \}$ from the permanent impact kernel $H$ from~\eqref{eq:permanent_impact_specification}: the smaller \(\kappa_{\infty}\) is, the smaller the permanent impact becomes, and consequently the more aggressive the resulting optimal trades are.

\begin{figure}[H]
\centering
\includegraphics[width=6in]{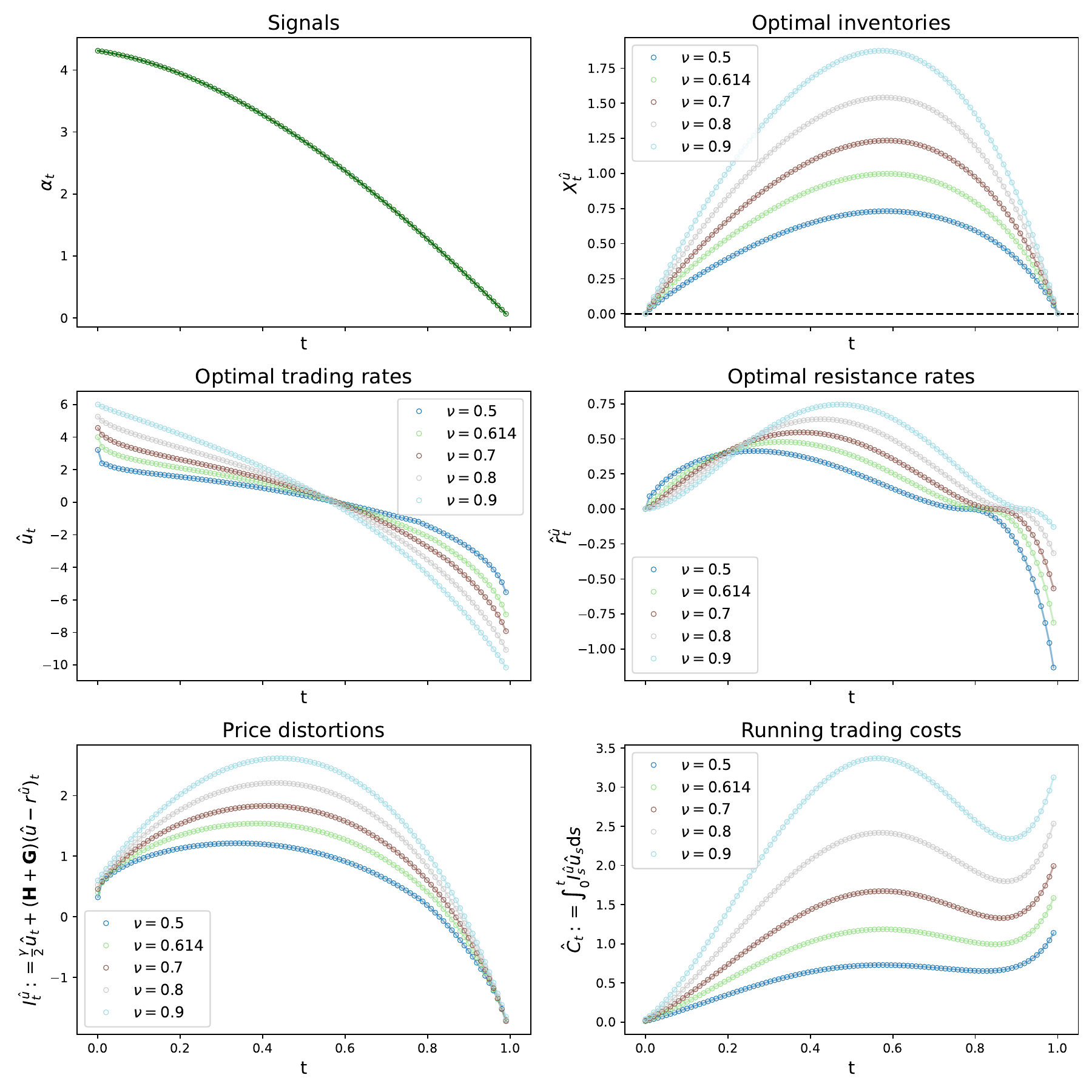}
\caption{Optimal round-trips with different impact decay parameters $\nu \in \{ 0.5, 0.614, 0.7, 0.8, 0.9 \}$ used in the power-law kernel $G_{\lambda, \nu}$ from~\eqref{eq:decay_specification}.}
\label{F:trading_buy_signals_various_impact_decay}
\end{figure}

\begin{figure}[H]
\centering
\includegraphics[width=6in]{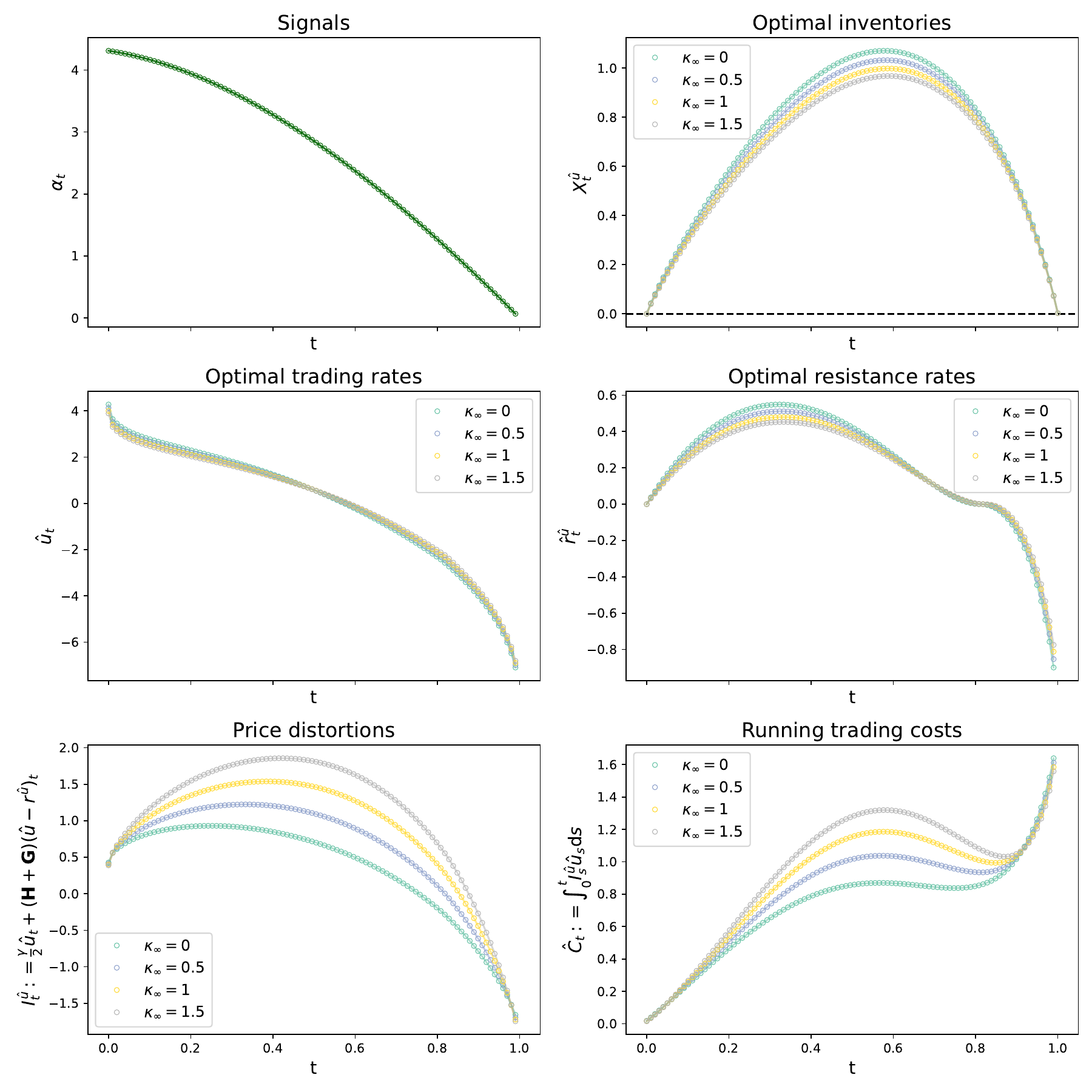}
\caption{Optimal round-trips with different permanent impact parameters $\kappa_\infty \in \{ 0.5, 1, 1.5 \}$ used in the permanent impact kernel $H$ from~\eqref{eq:permanent_impact_specification}.}
\label{F:trading_buy_signals_various_permanent_impact}
\end{figure}

\section{Mathematical tools}

\subsection{Completely monotone functions}

\begin{definition}[Completely monotone function]
\label{def:cm}
    A function $G:(0,\infty)\to\R$ is said to be completely monotone if $G\in C^\infty(0,\infty)$ and
    \begin{equation*}
        (-1)^n G^{(n)}(t)\ge 0,\qquad \forall\, t>0,\ \forall\, n \geq 0.
    \end{equation*}
    If, in addition, $G$ admits a (finite) right limit at $0$, we extend it to $[0,\infty)$ by setting
    $G(0):=\lim_{t\to 0^+}G(t)$.
\end{definition}

Completely monotone functions are, in particular, nonnegative, nonincreasing, and convex on $(0,\infty)$.
We refer to \citet{Bernstein1929} for further properties and historical background.

\begin{example}
\label{ex:cm}
On $[0,\infty)$, the following functions are completely monotone:
\begin{itemize}
    \item Exponentials: for any $\lambda\ge 0$, $t\mapsto e^{-\lambda t}$.
    \item Finite or countable sums (or mixtures) of exponentials with nonnegative weights:
    \begin{equation*}
        G(t)=\sum_{k\ge 1} a_k e^{-\lambda_k t},\qquad a_k\ge 0,\ \lambda_k\ge 0,
    \end{equation*}
    whenever the series converges (e.g.\ pointwise for all $t\ge 0$).
    \item Power laws: for any $\alpha>0$ and any $b \geq 0$, $t\mapsto (b+t)^{-\alpha}$.
\end{itemize}
\end{example}

The key structural result is that completely monotone functions are exactly Laplace transforms of positive measures.

\begin{theorem}[Bernstein-Widder representation]
\label{thm:bernstein_widder}
    Let $G:(0,\infty)\to[0,\infty)$ be completely monotone.
    Then there exists a (unique) $\sigma$-finite Borel measure $\mu$ on $[0,\infty)$ such that
    \begin{equation}
    \label{eq:cm_laplace_representation}
        G(t)=\int_{[0,\infty)} e^{-\lambda t}\,\mu(\dd\lambda),\qquad t>0.
    \end{equation}
    If moreover $G(0+)<\infty$, then $\mu$ is a finite measure and $\mu([0,\infty))=G(0+)$, so that
    \eqref{eq:cm_laplace_representation} holds for all $t\ge 0$ by defining $G(0):=G(0+)$.
\end{theorem}
\begin{proof}
    See Theorem~3 in \cite{lax2014functional}.
\end{proof}

\subsection{Operators in $\mathcal{L}^2$ and in $L^2$}

\begin{definition}
    We say that an operator $\mathbf{A}$ on $L^2$ is non-anticipative if for all $0 \leq t \leq T$ and all $u, v$ in $L^2$, we have 
    \begin{equation*}
        u \1_{[0,t]} = v\1_{[0,t]} \Rightarrow \mathbf{A} u(t) = \mathbf{A} v(t).
    \end{equation*}
\end{definition}

% \begin{rem}
%     \color{red}
%     This definition extends the definition of nonanticipative integral operators. Explain why
% \end{rem}

Non-anticipative operators on $L^2$ play a special role because they naturally induce operators on $\mathcal{L}^2$.

\begin{lemma}
    Suppose that $\mathbf{A}$ is a non-anticipative operator on $L^2$ such that there exists $C > 0$ such that
    \begin{equation*}
        \norm{\mathbf{A} u}_{L^2} \leq C \norm{u}_{L^2}.
    \end{equation*}
    For each $u \in \mathcal{L}^2$, we define 
    \begin{equation*}
        \mathbf{B} (u)(\omega, t) = \mathbf{A} (u(\omega, \cdot))(t).
    \end{equation*}
    Then $\mathbf{B}$ defines an operator on $\mathcal{L}^2$ such that
    \begin{equation*}
         \norm{\mathbf{B} u}_{\mathcal{L}^2} \leq C \norm{u}_{\mathcal{L}^2}.
    \end{equation*}    
    In that case, we write $\mathbf{B} = \widehat{\mathbf{A}}$.
\end{lemma}

\begin{lemma}
\label{lem:calL2_to_L2}
    Suppose that $\mathbf{B}$ is an operator on $\mathcal{L}^2$ for which there exists $C > 0$ such that
    \begin{equation*}
        \norm{\mathbf{B} u}_{\mathcal{L}^2} \leq C \norm{u}_{\mathcal{L}^2}
    \end{equation*}
    and suppose that if $u$ is a deterministic process, then $\mathbf{B} u$ is also deterministic. Then for each $u \in L^2$, we define 
    \begin{equation*}
        \mathbf{A}u(t) = \mathbf{B} u(t).
    \end{equation*}
    Then $\mathbf{A}$ defines an non-anticipative operator on $L^2$ such that
    \begin{equation*}
        \norm{\mathbf{A} u}_{L^2} \leq C \norm{ u }_{L^2}.
    \end{equation*}   
    In that case, we write $\mathbf{A} = \widecheck{\mathbf{B}}$.
    Moreover, we have $\widehat{\widecheck{\mathbf{B}}} = \mathbf{B}.$
\end{lemma}

\subsection{Fr\'echet differentiability}

\begin{definition}
\label{def:frechet:diff}
    Let $E$ and $F$ be two Banach spaces and $V$ be an open subset of $E$. An application $f: V \to F$ is said to be Fr\'echet differentiable at $u \in V$ if there exists a linear application $Df(u): E \to F$ such that the following limit holds
    \begin{equation*}
        \lim \limits_{\norm{h}_E \to 0} \frac{\norm{f(u+h) - f(u) - Df(u)(h)}_F}{\norm{h}_E} = 0.
    \end{equation*}
\end{definition}

\begin{theorem}[Implicit function theorem]
\label{thm:implicit}
    Let $X, Y, Z$ be Banach spaces. Let the mapping $f : X \times Y \to Z$ be continuously Fréchet differentiable in the sense of Definition~\ref{def:frechet:diff}.
    % that for all $(x, y) \in X \times Y$ and all $(h, g) \in X \times Y$, there exists a linear bounded operator $Df(x,y) : X \times Y \to Z$ such that
    % \begin{equation*}
    %     \lim_{\|(h,g)\|_{X \times Y} \to 0} \frac{\|f(x + h, y + g) - f(x , y) - Df(x,y) (h,g)\|_{X \times Y}}{\|(h,g)\|_{X \times Y}}.
    % \end{equation*}
    Then if $(x_0, y_0) \in X \times Y$ satisfies $f(x_0, y_0) = 0$ and if $y \mapsto Df(x_0, y_0)(0,y)$ is a Banach space isomorphism from $Y$ onto $Z$, then there exist neighborhoods $U$ of $x_0$ and $V$ of $y_0$ and a Fréchet differentiable function $g : U \to V$ such that $f(x, g(x)) = 0$ and $f(x, y) = 0$ if and only if $y = g(x)$, for all $( x , y ) \in U \times V$.
\end{theorem}

\section{Study of a nonlinear Volterra equation}
\label{app:nonlinear_volterra}

\subsection{Definition and main results.}

In this section, we are interested in the following nonlinear Volterra equation.

\begin{definition}
Let $k: [0,\infty) \to [0,\infty)$ be a locally integrable function  such that for any $t_0 > 0$, we have
\begin{equation*}
\lim\limits_{t \to t_0} \int_0^{\infty} |k(t_0 - s) - k(t-s)| \, \dd s \to 0.
\end{equation*}
For $T>0$, we consider $f: [0, T] \to \mathbb{R}$ and $g: \mathbb{R} \to \mathbb{R}$. We write $(E_{k,f,g})$ the Volterra integral Equation given by
\begin{equation*}
x(t) + \int_0^t k(t-s) g\big(x(s)\big) \, \dd s = f(t) 
\end{equation*}
where $x: [0,T] \to \mathbb{R}$.
\end{definition}

We first state the following existence and uniqueness result.

\begin{theorem}
\label{thm:existence:E}
    Suppose that $g$ is Lipschitz continuous and let $T > 0$. Then, for every locally bounded function $f$ on $[0, T]$, there exists a unique solution to $(E_{k,f,g})$. Moreover, there exists a constant $C_T > 0$, depending only on the Lipschitz constant of $g$, on $T$, and on the kernel $k$, such that for any two bounded functions $f_1$ and $f_2$ on $[0, T]$, if $x_1$ and $x_2$ denote the corresponding solutions to $(E_{k,f_1,g})$ and $(E_{k,f_2,g})$, respectively, then we have
    \begin{equation}
    \label{eq:lip:E}
        \norm{x_1-x_2}_{\infty} \leq C_T \norm{f_1-f_2}_{\infty}.
    \end{equation}
\end{theorem}

For a given Lipschitz continuous function $g$, Theorem~\ref{thm:existence:E} implies that there exists an operator $\mathbf{S}_T : L^\infty([0,T]) \to L^\infty([0,T])$ such that for every locally bounded function $f$ on $[0,T]$, $x = \mathbf{S}_T f$ is the unique solution of $(E_{k,f,g})$. By Theorem~\ref{thm:existence:E}, we already know that this functional is Lipschitz continuous in $L^\infty([0,T])$. This means that there exists $c_T > 0$ such that 
\begin{equation*}
    c_T \norm{x_1 - x_2}_{L^\infty} \leq \norm{\mathbf{S}_T^{-1} x_1 - \mathbf{S}_T^{-2} x_2}_{L^\infty}.
\end{equation*}

\subsection{Proof of Theorem~\ref{thm:existence:E}}

Let $g$ be a Lipschitz continuous function, $L$ its Lipschitz constant and let $\varepsilon > 0$ be small enough so that 
\begin{equation*}
    \int_0^\varepsilon k(s) \, \dd s \leq \frac{1}{2L}.
\end{equation*}
The proof is split in two parts. First, we prove Theorem~\ref{thm:existence:E} on $[0, \varepsilon]$. Then, we show that if Theorem~\ref{thm:existence:E} is true on $[0,T]$, then it also holds on $[0, T+\varepsilon]$. Since $\bbR$ is Archimedean, we can then conclude by induction.\\

\textbf{Step 1.} 
We define
\begin{equation*}
\mathcal{T}_0 x(t) = \int_0^t k(t-s) g\big(x(s)\big) \, \dd s
\end{equation*}
for any function $x$ in $L^{\infty}([0,\varepsilon])$ and all $0 \leq t \leq \varepsilon$. This defines an operator $\mathcal{T}_0$ on $L^{\infty}([0,\varepsilon])$. Moreover, for all bounded functions $x_0$ and $x_1$ on $[0, \varepsilon]$, we have
\begin{align*}
|\mathcal{T}_0 x_1(t) - \mathcal{T}_0 x_2(t)| 
&\leq \int_0^t k(t-s) | g\big(x_1(s)\big) - g\big(x_2(s)\big)| \, \dd s
\\
&\leq L \int_0^t k(t-s) | x_1(s) - x_2(s)| \, \dd s
\\
&\leq L \int_0^\varepsilon k(s) \, \dd s \norm{x_1 - x_2}_{\infty}\\
&\leq \frac{1}{2} \norm{x_1 - x_2}_{\infty}.
\end{align*}
Thus, if $f$ is a bounded function on $[0, \varepsilon]$, the operator $\wt {\mathcal{T}}_0$ defined on $L^\infty([0, \varepsilon])$ by
\begin{equation*}
\wt {\mathcal{T}}_0 x(t) := f(t) -  {\mathcal{T}}_0 x(t)
\end{equation*}
is a contraction of $L^{\infty}([0,\varepsilon])$ and therefore admits a unique fixed point $x \in L^{\infty}([0,\varepsilon])$, which is also a solution of $(E_{k,f,g})$ on $[0, \varepsilon]$. This ensures that $I + \mathcal{T}_0$ is invertible, where $I$ stands here for the identity operator on $L^\infty([0, \varepsilon])$,  and therefore that $x = (I + \mathcal{T}_0)^{-1} f$.

We now prove that $(I + \mathcal{T}_0)^{-1}$ is $2$-Lipschitz continuous, which implies~\eqref{eq:lip:E} with $T = \varepsilon$. We consider $f_1, f_2$ two bounded functions on $[0, \varepsilon]$ and we write $x_1 = (I + \mathcal{T}_0)^{-1} f_1$ and $x_2 = (I + \mathcal{T}_0)^{-1} f_2$. We then have
\begin{align*}
    \norm{f_1 - f_2}_\infty 
    &= \norm{(I + \mathcal{T}_0)x_1 - (I + \mathcal{T}_0)x_2}_\infty 
\\
    &\geq \norm{x_1 - x_2}_\infty  - \norm{\mathcal{T}_0 x_1 - \mathcal{T}_0x_2}_\infty.
\end{align*}
We then obtain
\begin{align*}
    \norm{x_1 - x_2}_\infty &\leq \norm{f_1 - f_2}_\infty  + \norm{\mathcal{T}_0 x_1 - \mathcal{T}_0x_2}_\infty 
    \\
    &\leq \norm{f_1 - f_2}_\infty  + \frac{1}{2} \norm{x_1 - x_2}_\infty,
\end{align*}
and thus
\begin{equation*}
     \norm{x_1 - x_2}_\infty \leq 2\norm{f_1 - f_2}_\infty
\end{equation*}
which concludes the proof of the first step. \\

\textbf{Step 2.}
We now assume the result holds on $[0, T]$ and want to extend it on $[0, T+\varepsilon]$. To this end, we introduce the operator $\mathcal{T}$ acting on $L^\infty([T, T+\varepsilon])$ functions as follows
\begin{equation*}
\mathcal{T} x(t) = \int_T^t k(t-s) g\big(x(s)\big) \, \dd s, \quad T \le t \le T+\varepsilon.
\end{equation*}
Proceeding as previously, this operator in $\frac{1}{2}$-Lipschitz and $(I+\mathcal{T})^{-1}$ is invertible, where $I$ stands here for the identity operator on $L^\infty([T, T+\varepsilon])$. Moreover, we can also prove that $(I+\mathcal{T})^{-1}$ is $2$-Lipschitz continuous on $L^\infty([T, T+\varepsilon])$.

Now, let $f$ be a bounded function on $[0,T+\varepsilon]$. By assumption, we know that $(E_{k,f,g})$ admits a solution $x_T$ on $[0,T]$. Consider $\wt f \in L^\infty([T, T+\varepsilon])$ defined by
\begin{equation*}
    \wt f(t) = f(t) - \int_0^T k(t-s) g\big(x_T(s)\big) \, \dd s, \quad T \le t \le T + \varepsilon.
\end{equation*}
Setting $\wt x = (I+\mathcal{T})^{-1} \wt f$, we have for all $T \leq t \leq T+\varepsilon$
\begin{equation*}
    \wt x(t) + \int_T^t k(t-s) g\big(x(s)\big) \, \dd s = f(t) - \int_0^T k(t-s) g\big(x_T(s)\big) \, \dd s.
\end{equation*}
Since $x_T$ is a solution of $(E_{k,f,g})$ on $[0,T]$, we have that
$f(T) - \int_0^T k(T-s) g(x_T(s)) \, \dd s = x_T(T).$
We can then define without ambiguity 
\begin{equation*}
x(t) = 
\begin{cases}
x_T(t) & \text{ if } t \leq T,
\\
\wt x(t) & \text{ if } t \geq T
\end{cases}
\end{equation*}
which is a continuous function and is, by construction, a solution of $(E_{k,f,g})$ on $[0, T + \varepsilon]$.

It remains to prove~\eqref{eq:lip:E}. Consider $f_1, f_2$ two bounded functions on $[0, T+\varepsilon]$ and let $x_1$ and $x_2$ be the solutions of $(E_{k,f_1,g})$ and $(E_{k,f_2,g})$ respectively. Therefore, restricting the solutions on $[0,T]$, we know from the initial assumption that there exists a positive constant $C_T$ such that
\begin{equation*}
    \sup_{t \leq T} |x_1(t) - x_2(t)| \leq C_T \sup_{t \leq T} |f_1(t) - f_2(t)|.
\end{equation*}
Now, we define $\wt f_1$ and $\wt f_2$ for $t \in [T, T + \varepsilon]$ by
\begin{align*}
    \wt f_1(t) &= f_1(t) - \int_0^T k(t-s) g\big(x_1(s)\big) \, \dd s,
    \\
    \wt f_2(t) &= f_2(t) - \int_0^T k(t-s) g\big(x_2(s)\big) \, \dd s
\end{align*}
so that 
\begin{equation*}
\begin{split}
    (x_1(t))_{T \leq t \leq T+\varepsilon} &= (I+\mathcal{T})^{-1} \wt f_1,
\\
    (x_2(t))_{T \leq t \leq T+\varepsilon} &= (I+\mathcal{T})^{-1} \wt f_2.
\end{split}
\end{equation*}
We then conclude using that $(I+\mathcal{T})^{-1}$ is $2$-Lipschitz and that
\begin{align*}
    |\wt f_1(t) - \wt f_2(t)| 
    &
    \leq |f_1(t) - f_2(t)| + 
    \bigg| \int_0^T k(t-s) g\big(x_1(s)\big) \dd s - \int_0^T k(t-s) g\big(x_2(s)\big)  \dd s\bigg|
\\
    &
    \leq |f_1(t) - f_2(t)| + 
    \bigg| \int_0^T k(t-s) g\big(x_1(s)\big) \, \dd s - \int_0^T k(t-s) g\big(x_2(s)\big)  \dd s\bigg|
\\
    &
    \leq \norm{f_1 - f_2}_{\infty} + L \norm{x_1 - x_2}_{\infty} \bigg| \int_0^T k(t-s) \dd s\bigg|
\\
    &
    \leq \norm{f_1 - f_2}_{\infty} \bigg(1 + L \int_0^T k(t-s) \dd s \bigg).
\end{align*}

\section{Properties of the market resistance $r^u$}

\subsection{Proof of Lemma~\ref{lem:well_posed}}
\label{app:proof:lem:well_posed}

We first state the following lemma which follows directly from Theorem~\ref{thm:existence:E}.

\begin{lemma}
\label{lem:preliminary}
    Suppose that $\calU$ is Lipschitz continuous with linear growth. For each $T \geq 0$, there exists an operator $\mathbf{R}_T: \mathcal{C}([0, T]) \to \mathcal{C}([0, T])$ such that for each $f : [0, T] \to \bbR$, the function $r(s) = (\mathbf{R}_T f)(s)$ is the unique solution of
    \begin{equation}
    \label{eq:volterra:r}
    r(t)  = \calU \bigg( f(t) - \int_0^t G(t-s) r(s) \dd s \bigg).
    \end{equation}
    Moreover, $\mathbf{R}_T$ is Lipschitz continuous on $\mathcal{C}([0, t])$ and its Lipschitz constant depends only on $G$, $T$ and on the Lipschitz constant of $\calU$.
\end{lemma}
\begin{proof}
    First note that if $r$ is a solution of~\eqref{eq:volterra:r}, then $x(t) = f(t) - \int_0^t G(t-s) r(s) \dd s$ is solution of $(E_{k,f,\calU})$. Therefore Theorem~\ref{thm:existence:E} ensures that the solution of~\eqref{eq:volterra:r} is unique. Moreover, if $x$ is a  solution of $(E_{k,f,\calU})$, then $r(t) = \calU(x(t))$ is clearly a solution of~\eqref{eq:volterra:r}. This ensures that we can define an operator $\mathbf{R}_T: \mathcal{C}([0, T]) \to \mathcal{C}([0, T])$ as prescribed in Lemma~\ref{lem:preliminary}. The Lipschitz property of $\mathbf{R}_T$ follows directly from~\eqref{eq:lip:E}.
\end{proof}

To use this result, note that since $u \in \mathcal{L}^2$, the function $u_{.}(\omega)$ is in $L^2$ for almost all $\omega \in \Omega$. Fix now $\omega \in \Omega$ such that it is the case. Then we know that $f(\omega) = \mathbf{G} u(\omega)$ is a continuous function and thus using Lemma~\ref{lem:preliminary}, we can define $r^u(\omega) = \mathbf{R}_T f(\omega)$ which is solution of
\begin{equation*}
    r^u_t(\omega) = \calU \Big(\mathbf{G}\big(u(\omega)-r^{u}(\omega)\big)_{t}\Big).
\end{equation*}
Since this is done for almost all $\omega$, this allows us to build a stochastic process $r^u$. Note that by construction, this process is almost surely continuous. Moreover, it is also adapted, seeing that for each $t \in [0,T]$, we have
\begin{equation*}
    \big(r^u_s(\omega)\big)_{0 \leq s \leq t} = \mathbf{R}_t f(\omega) = \mathbf{R}_t \mathbf{G} u(\omega)
\end{equation*}
and $\mathbf{R}_t \mathbf{G} u(\omega)$ is $\calF_t$-measurable because $(u_s)_{s\leq t}$ is $\calF_t$-measurable and $\mathbf{G}$ is non-anticipative. Thus $r^u$ is adapted and continuous, and therefore it is progressively measurable. 

It remains to prove that this process belongs to $\mathcal{L}^2$. First, note that $\mathbf{R}_T 0 = 0$ because $\calU(0) =0$. Then, the Lipschitz property of $\mathbf{R}_T$ ensures that
\begin{equation*}
    \sup_{t \leq T} |r^u_t| \leq C \sup_{t \leq T} |\mathbf{G} u_t| 
\end{equation*}
for some constant $C$ depending only on the Lipschitz constant of $\calU$, $T$ and on $G$. Using Young's convolution inequality, we obtain
\begin{equation*}
    \sup_{t \leq T}  |\mathbf{G} u_t| \leq \norm{G}_{L^2([0,T])} \norm{u}_{L^2([0,T])}.
\end{equation*}
and therefore \begin{equation*}
    \E \Big[ \norm{r^u}_{L^2}^2 \Big] 
    \leq C^2 T \norm{G}_{L^2([0,T])}^2 \E[\norm{u}_{L^2}^2]
\end{equation*}
which is finite because $G \in L^2([0,T])$ and $\E[\norm{u}_{L^2}^2] < \infty$ as $u \in \mathcal{L}^2$.

\subsection{Proof of Lemma~\ref{lem:frechet:R}} \label{ss:proof_lemma_differentiability}

The proof of Lemma~\ref{lem:frechet:R} is based on the implicit function theorem recalled in Theorem~\ref{thm:implicit}. We apply this to study the Fréchet differentiability of $u \mapsto r^u$. For each $u \in \mathcal{L}^2$, recall Lemma~\ref{lem:well_posed} yields the existence and uniqueness of $r^u \in \mathcal{L}^2$ such that
\begin{equation*}
    r_{t}^{u} = \calU \Big( \big(\mathbf{G}_{2}(u-r^{u})\big)_{t} \Big),
\end{equation*}
defining unambiguously the operator $\mathbf{R}$ such that $r^u = \mathbf{R}(u)$. Let us introduce the operator $\mathbf{T} : \mathcal{L}^2 \times \mathcal{L}^2 \to \mathcal{L}^2$ by
\begin{equation*}
    \mathbf{T}(u,r) := r - \calU \big(\mathbf{G}(u-r)\big), \quad u,r \in \mathcal{L}^{2},
\end{equation*}
so that we always have
\begin{equation}
\label{eq:TR0}
    \mathbf{T}\big(u,\mathbf{R}(u)\big) = 0, \quad u \in \mathcal{L}^{2}.
\end{equation}
We want to apply the implicit function theorem to the operator $\mathbf{T}$ to get the Fréchet differentiability of $\mathbf{R}$. First we need to check that $\mathbf{T}$ is Fréchet differentiable, which is the case, since $\calU$ is differentiable, by composition of Fréchet differentiable functions. Moreover, explicit computations show that 
\begin{equation*}
    D\mathbf{T}(u,r)(s, v) = v - \calU' (\mathbf{G}(u-r))\mathbf{G}(s-v), \quad s,v \in \mathcal{L}^{2}.
\end{equation*}
Finally, when $u$ is fixed in $\mathcal{L}^2$, then 
\begin{equation*}
    v \mapsto D\mathbf{T}\big(u,\mathbf{R}(u)\big)(0, v) = v + \calU' \Big(\mathbf{G}\big(u-\mathbf{R}(u)\big)\Big) \mathbf{G} v
\end{equation*}
which is a Banach space isomorphism from $\mathcal{L}^2$ onto $\mathcal{L}^2$. Concretely, we check that for each process $w \in \mathcal{L}^2$, there exists a process $y \in \mathcal{L}^2$ such that 
\begin{equation}
\label{eq:goal:vw}
    y + \calU' (\mathbf{G}(u-r^u)) \mathbf{G} y = w.
\end{equation}
The existence and uniqueness of a stochastic process $v$ satisfying~\eqref{eq:goal:vw} is guaranteed by the following lemma.
\begin{lemma}
    \label{lemma:solution:linear_volterra}
    For each $T \geq 0$, there exists a continuous operator $\mathbf{Y}_T: \calC([0,T]) \times L^2([0,T]) \to L^2([0,T])$ such that for each continuous function $f$ on $[0,T]$ and each $g$ in $L^2([0,T])$, $y(t) = \mathbf{Y}_T(f,g)(t)$ is solution of
    \begin{equation}
    \label{eq:volterra_product}
        y(t) + f(t) \mathbf{G} y(t) = g(t).
    \end{equation}
    Moreover, the solution $y$ of~\eqref{eq:volterra_product} is unique and there exists a deterministic function $\psi : [0, \infty) \to [0, \infty) $ such that we have
\begin{equation}
\label{eq:bound:solution:linear_volterra}
    \norm{y}_{L^2([0,T])} \leq \psi(\norm{f}_{L^{\infty}} )\norm{g}_{L^2([0,T])}.
\end{equation}
\end{lemma}
Note that $\mathbb{P}$-almost all $\omega$, the function $f(t) = \calU' (\mathbf{G}(u-r^u)(\omega, t))$ is continuous on $[0,T]$, bounded by $\|\mathcal{U}\|_{L^{\infty}}$, and $g = w(\omega, .)$ is well in $L^2([0,T])$. Therefore, the stochastic process $v$ defined on $[0,T]$ by
\begin{equation*}
    v(\omega, t) = \mathbf{Y}_T \bigg( \calU' \Big(\mathbf{G}\big(u(\omega, \cdot)-\mathbf{R}(u)(\omega, \cdot)\big)\Big), w(\omega, \cdot)\bigg)(t)
\end{equation*}
is the unique solution of~\eqref{eq:bound:solution:linear_volterra} for the specified function $f$ and $g$. Furthermore, for each $t$, $v(\cdot, t)$ is measurable as the composition of measurable mappings. Using the same arguments as in Section~\ref{app:proof:lem:well_posed}, we see that the stochastic process $v(\omega, \cdot)$ is adapted. In fact, for each $0 \leq t \leq T$, we have  
\begin{equation*}
    \big( v(\omega, s) \big)_{0\leq s \leq t} = \mathbf{Y}_t \bigg( \Big( 
    \calU' \big(\mathbf{G}u(\omega, s)-\mathbf{G}\mathbf{R}(u)(\omega, s)\big)\Big)_{0\leq s \leq t}, 
    \big( w(\omega, s) \big)_{0\leq s \leq t}
    \bigg )
\end{equation*}
by uniqueness of the solution of~\eqref{eq:volterra_product}.
The right-hand side of this identity is $\mathcal{F}_t$-measurable because $\mathbf{G}u$, $\mathbf{G}\mathbf{R}(u)$ and $w$ are adapted, and because $\mathbf{Y}_t$ is a continuous application. This means that the left hand side is also $\mathcal{F}_t$-measurable, and in particular $v(\omega, t)$ is $\mathcal{F}_t$-measurable, implicating that $v$ is adapted, Moreover, since it is continuous almost surely by construction, it is also progressively measurable. Using~\eqref{eq:bound:solution:linear_volterra}, we know that
\begin{equation*}
    \norm{v}_{L^2} \leq \psi(\norm{\calU'}_{L^{\infty}} )\norm{w}_{L^2}
\end{equation*}
almost surely, and since $\calU'$ is deterministic and bounded we get
\begin{equation*}
    \norm{v}_{\calL^2} \leq \psi(\norm{\calU'}_{L^{\infty}} )\norm{w}_{\calL^2}
\end{equation*}
which ensures that $v \in \calL^2$. Finally, the implicit function theorem from Theorem~\ref{thm:implicit} applies, which shows that $\mathbf{R}$ is Fr\'echet differentiable. It remains to explicitly compute $D \mathbf{R}$ to conclude the proof of Lemma~\ref{lem:frechet:R}. Differentiating~\eqref{eq:TR0}, we see that $D \mathbf{R}$ must satisfy
\begin{equation*}
    D\mathcal{T}\big(u, \mathbf{R}(u)\big)\big(v, D\mathbf{R}(v)\big) = 0,
\end{equation*}
and therefore
\begin{equation*}
    D\mathbf{R}(u) (v) =  \calU' \Big(\mathbf{G}\big(u-\mathbf{R}(u)\big)\Big)\mathbf{G}\big(v - D\mathbf{R}(u)(v)\big).
\end{equation*}
Finally, we have
\begin{align*}
    D\mathbf{R}(u) 
    &= \bigg(\mathbf{I} + \calU' \Big(\mathbf{G}\big(u-\mathbf{R}(u)\big)\Big)\mathbf{G}\bigg)^{-1} \calU' \Big(\mathbf{G}\big(u-\mathbf{R}(u)\big)\Big)\mathbf{G}
    \\
    &= \mathbf{I} - \bigg(\mathbf{I} + \calU' \Big(\mathbf{G}\big(u-\mathbf{R}(u)\big)\Big)\mathbf{G}\bigg)^{-1}.
\end{align*}

\subsection{Poof of Lemma~\ref{lemma:solution:linear_volterra}}

The proof of Lemma~\ref{lemma:solution:linear_volterra} is split into two parts. In the first part, we fix a continuous function $f$ on $[0,T]$ and prove that~\eqref{eq:volterra_product} admits a solution $y$ for each continuous function $g$ on $[0,T]$ and that the mapping $g \mapsto y$ is Lipschitz continuous with a Lipschitz constant that only depends on $\norm{f}_{L^\infty}$. In the second part of the proof, we show this can be used to conclude that the mapping $\mathbf{Y}_T$ is continuous.\\

\textbf{Step 1.} Suppose that $f$ is a continuous function on $[0,T]$. We introduce $\varepsilon > 0$ small enough so that
\begin{equation*}
    \norm{f}_{L^2} \int_0^\varepsilon G(s) \, \dd s \leq \frac{1}{2}.
\end{equation*}
Similarly to the proof of Theorem~\ref{thm:existence:E}, we can prove that the equation~\eqref{eq:volterra_product} can be solved uniquely on $[0,\varepsilon]$. Moreover, the mapping $\mathbf{V}_\varepsilon^f$ defined such that for each continuous function $g$, $y = \mathbf{V}_\varepsilon^f g$ is the unique solution of~\ref{thm:existence:E}, is $2$-Lipschitz continuous on $L^2([0, \varepsilon])$. Following again the proof of Theorem~\ref{thm:existence:E}, we see that we can then extend this solution to a solution on $[0, T]$. Furthermore, if we define $\mathbf{V}_t^f$ as $\mathbf{V}_\varepsilon^f$ on $L^2([0, t])$ instead of $L^2([0, \varepsilon])$, we see that $\mathbf{V}_t^f$ is $L_t$-Lipschitz continuous on $L^2([0, t])$ for some $L_t > 0$. Repeating the computations of the proof of Theorem~\ref{thm:existence:E}, we can check that for each $0 \leq t \leq T-\varepsilon$, we have
\begin{equation*}
    L_{t+\varepsilon} \leq 2 \bigg(1 + \norm{f}_\infty \int_0^T G(s) \dd s L_t\bigg).
\end{equation*}
By induction, we deduce that
\begin{equation*}
    L_T \leq 2 \sum_{k=0}^{n} b^k
\end{equation*}
where $b = 2 \norm{f}_\infty \int_0^T G(s) \dd s$ and $n=\lfloor T / \varepsilon \rfloor$. Therefore, we deduce that there exists an increasing function $\psi$ such that $\mathbf{V}_T^f$ is $\psi(\norm{f}_\infty)$-Lipschitz continuous. In particular, we obtain that for each continuous function $g$, we have
\begin{equation*}
    \norm{\mathbf{V}_T^f g}_{L^2([0,T])} \leq \psi(\norm{f}_{L^\infty([0,T])}) \norm{g}_{L^2([0,T])}
\end{equation*}
because $\mathbf{V}_T^f (0) = 0$, which proves~\eqref{eq:bound:solution:linear_volterra}.\\

\textbf{Step 2.} Using the results of the first step, we see that we can always take
\begin{equation*}
    \mathbf{Y}_T(f,g) = \mathbf{V}_T^f g.
\end{equation*}
It remains to prove that this application is continuous. To do this, we prove that for $M > 0$, $\mathbf{Y}_T$ is Lipschitz-continuous on $\mathcal{B}_{\mathcal{C}([0,T])}(0, M) \times \mathcal{B}_{L^2([0,T])}(0, M)$ where $\mathcal{B}_{E}(x, \epsilon)$ denotes the ball in $E$ centered in $x$ with radius $\epsilon$. From the first step, we know that if $f \in \mathcal{B}_{\mathcal{C}([0,T])}(0, M)$, the application $g \mapsto \mathbf{Y}_T(f,g)$ is $\psi(M)$-Lipschitz continuous. Now let $f_1$ and $f_2$ be two continuous functions bounded by $M$ and let $g \in L^2([0,T])$ with $\norm{g}_{L^2([0,T])} \leq M$. We write $y_i = \mathbf{Y}_T(f_i,g)$. By definition, we have for all $0 \leq t \leq T$ 
\begin{equation*}
    y_2(t) + f_2(t) \mathbf{G}y_2(t) = g(t)
\end{equation*}
which is equivalent to
\begin{equation*}
    y_2(t) + f_1(t) \mathbf{G}y_2(t) = g(t) - (f_2(t)-f_1(t)) \mathbf{G}y_2(t).
\end{equation*}
Therefore, we have
\begin{equation*}
    y_2 = \mathbf{Y}_T(f_1,g - (f_2-f_1) \mathbf{G}y_2).
\end{equation*}
This implies in particular that
\begin{align*}
    \norm{y_2 - y_1}_{L^2([0,T])} 
    &\leq
    \norm{\mathbf{Y}_T(f_1,g - (f_2-f_1) \mathbf{G}y_2) - \mathbf{Y}_T(f_1,g)}_{L^2([0,T])} 
\\
    &\leq
    \psi(M) \norm{(f_2-f_1) \mathbf{G}y_2)}_{L^2([0,T])}
    \\
    &\leq
    \psi(M) \norm{f_2-f_1}_{L^\infty([0,T])} \norm{\mathbf{G}}_{L^2([0,T])}
    \norm{y_2}_{L^2([0,T])}.
\end{align*}
Therefore, we have that 
\begin{equation*}
    \norm{\mathbf{Y}_T(f_2,g) - \mathbf{Y}_T(f_1,g)}_{L^2([0,T])} 
    \leq
    C \norm{f_2-f_1}_{L^\infty([0,T])} 
\end{equation*}
with $C = M \psi(M) \norm{\mathbf{G}}_{L^2([0,T])}$ and thus $f \mapsto \mathbf{Y}_T(f,g)$ is also Lipschitz continuous on $\mathcal{B}_{\mathcal{C}([0,T])}(0, M)$. Since both coordinates are Lipschitz on bounded sets, we deduce that $\mathbf{Y}_T$ is jointly continuous on $\mathcal{C}([0,T]) \times L^2([0,T])$, which completes the proof.

\section{Proof of Theorem~\ref{thm:linear_case}} 
\label{s:proof_linear_case}

First, using the linearity of $\calU$, Equation~\eqref{eq:r_u} rewrites 
\begin{equation*}
    r^u + a \mathbf{G} r^u = a \mathbf{G} u
\end{equation*}
and therefore 
\begin{equation*}
     r^u = \mathbf{R}(u) = a (\mathbf{I} + a \mathbf{G})^{-1} \mathbf{G} u.
\end{equation*}
Note that the invertibility of $\mathbf{I} + a \mathbf{G}$ is guaranteed by the existence and uniqueness of $r^u$ from Lemma~\ref{lem:well_posed}, or alternatively by a direct application of Theorem~\ref{thm:existence:E}.
Now, algebraic computations ensure that
\begin{equation*}
    \mathbf{I} - \mathbf{R} = \mathbf{I} - a (\mathbf{I} + a \mathbf{G})^{-1} \mathbf{G} = (\mathbf{I} + a \mathbf{G})^{-1},
\end{equation*}
so that the gain functional~\eqref{eq:operator:J} rewrites
\begin{equation*}
\mathcal{J}(u) = 
    \langle u, \alpha \rangle - \frac{\gamma}{2} \|u\|^2 - \langle u, (\mathbf{H} + \mathbf{G}) \circ (\mathbf{I} + a \mathbf{G})^{-1}u \rangle - \frac{\phi}{2} \norm{\bfX u}^{2}
    - \frac{\varrho}{2} \E[(\bfX u)_T^2] + X_0 \mathbb{E}[S_T].
\end{equation*}
Coercivity of the gain functional readily follows from the positive semi-definiteness of the operator $(\mathbf{H} + \mathbf{G}) \circ (\mathbf{I} + a \mathbf{G})^{-1}$, which is guaranteed hereafter by Lemma~\ref{L:psd_linear_resistance}. Moreover, to prove the $\gamma-$strong convexity of $-\mathcal{J}$, we prove equivalently the convexity of the functional
\begin{equation*}
    \widetilde{\mathcal{J}}(u) := - \mathcal{J}(u) - \frac{\gamma}{2} \norm{u}^2,\quad u \in \mathcal{L}^{2},
\end{equation*}
see \citet[Proposition 10.8]{bauschke2017}.
\newpage
By straightforward calculus, $\mathcal{J}$ is Gâteaux differentiable at any $u \in \mathcal{L}^{2}$, with Gâteaux gradient
\begin{align*} 
    \nabla \mathcal{J}(u) = &\alpha - X_0 \big(\phi(T- .) + \varrho\big) - \gamma u - \big((\mathbf{H} + \mathbf{G})\circ (\mathbf{I} + a \mathbf{G})^{-1} + (\mathbf{I} + a \mathbf{G}^*)^{-1} \circ (\mathbf{H}^* + \mathbf{G}^*)  \big) u\\
    & - (\mathbf{H}_{\phi,\varrho}
    + \mathbf{H}_{\phi,\varrho}^{*}) u.
\end{align*}
Therefore
\begin{equation*}
    \nabla \widetilde{\mathcal{J}}(u) = - \alpha - X_0 \big(\phi(T- .) + \varrho\big) + \big((\mathbf{H} + \mathbf{G})\circ (\mathbf{I} + a \mathbf{G})^{-1} + (\mathbf{I} + a \mathbf{G}^*)^{-1} \circ (\mathbf{H}^* + \mathbf{G}^*) \big) u + (\mathbf{H}_{\phi,\varrho}
    + \mathbf{H}_{\phi,\varrho}^{*}) u
\end{equation*}
By \citet[Proposition 17.10]{bauschke2017}, it is sufficient to show the monotonicity of $\nabla\widetilde{\mathcal{J}}$ to prove the convexity of $\widetilde{\mathcal{J}}$, i.e.,
\begin{equation} \label{eq:monotonicity_tilde_J}
    \langle u-v, \nabla \widetilde{\mathcal{J}}(u)-\nabla \widetilde{\mathcal{J}}(v)\rangle\ge0,\quad u,\, v\in \mathcal{L}^{2}.
\end{equation}
By linearity, property~\eqref{eq:monotonicity_tilde_J} is equivalent to the positive semi-definiteness of the operators $\mathbf{H}_{\phi,\varrho}$ and $(\mathbf{H} + \mathbf{G}) \circ (\mathbf{I} + a \mathbf{G})^{-1}$ which are proven in \citet[Lemma 4.3]{abi2025fredholm} and Lemma~\ref{L:psd_linear_resistance} respectively. Consequently, \citet[Theorem 4.1 (ii)]{abi2025fredholm} yields the existence and uniqueness of the optimal trading strategy $\hat{u}$ satisfying~\eqref{eq:optimal_trading_problem} and the first-order condition $\nabla \mathcal{J}(\hat{u}) = 0$ stated in~\eqref{eq:FOC:linear}. Finally, to complete the argument, we prove the following lemma.
\begin{lemma} \label{L:psd_linear_resistance}
    The operator $(\mathbf{H} + \mathbf{G}) \circ (\mathbf{I} - \mathbf{R})$ is positive semi-definite, i.e.,
    \begin{equation*}
        \langle u, \, (\mathbf{H} + \mathbf{G}) \circ (\mathbf{I} - \mathbf{R}) u \rangle \geq 0, \quad u \in \mathcal{L}^2.
    \end{equation*}
\end{lemma}
\begin{proof}
    Fix $u \in \mathcal{L}^{2}$. Since $\mathbf{I} - \mathbf{R} = (\mathbf{I} + a \mathbf{G})^{-1}$, we denote $v = (\mathbf{I} + a \mathbf{G})^{-1} u$ and we prove that
\begin{equation*}
    \langle (\mathbf{I} + a \mathbf{G}) v, \, (\mathbf{H} + \mathbf{G}) v \rangle \geq 0.
\end{equation*}
Expanding the inner product gives
\begin{equation*}
\langle (\mathbf{I} + a \mathbf{G}) v, \, (\mathbf{H} + \mathbf{G}) v \rangle 
 = \langle v, \boldsymbol{H}v\rangle + \langle v, \boldsymbol{G}v\rangle + a\norm{\boldsymbol{G}v}_{L^2([0,T])}^2 + a\langle \boldsymbol{G}v, \boldsymbol{H}v\rangle.
\end{equation*}
By the positive semi-definite property of $\boldsymbol{H}$ and $\boldsymbol{G}$, the first three terms are nonnegative. It remains then to show that
\begin{equation}
\label{eq:goal:24}
\langle \boldsymbol{H}v, \boldsymbol{G}v\rangle = \langle \boldsymbol{V}v,v\rangle \geq 0
\end{equation}
where
\begin{equation*}
% \boldsymbol{K} := \frac{\boldsymbol{V} + \boldsymbol{V}^*}{2}
% \quad \text{ and } \quad
\boldsymbol{V} := \boldsymbol{G}^* \circ \boldsymbol{H}.
\end{equation*}
We first prove~\eqref{eq:goal:24} holds for when $G$ is an exponential kernel, and extend the result to all admissible kernels as specified in Assumption~\ref{assumption:G}.

\paragraph{Exponential kernel case.} Let $\lambda > 0$ and assume that $G = G_\lambda$ given by 
\begin{equation} \label{eq:exp_kernel_for_proof}
G_\lambda(t) := e^{-\lambda t}, \quad t \geq 0.
\end{equation}
Applying Fubini, we obtain for any $u \in L^{2}([0,T])$
\begin{align*}
    \boldsymbol{V} u(t) &= ( \boldsymbol{G}^* \circ \boldsymbol{H} u )(t) = \int_{0}^{t} \frac{\kappa_{\infty}}{\lambda} u(r) \dd r + \int_{t}^{T} \frac{\kappa_{\infty}}{\lambda} e^{-\lambda(r-t)} u(r) \dd r - \int_{0}^{T} \frac{\kappa_{\infty}}{\lambda} e^{-\lambda(T-t)} u(r) \dd r \\
    &= \frac{\kappa_{\infty}}{\lambda}\big(\mathbf{1}u(t) + \boldsymbol{G}^*u(t) - e^{-\lambda (T-t)} \mathbf{1}u(T)\big)
\end{align*}
where $\mathbf{1}$ is the operator given by $\mathbf{1}u(t) = \int_0^t u(s) \dd s$ for every $u \in L^2([0,T])$ and $0 \le t \le T$. We then have 
\begin{align*}
    \langle \boldsymbol{V}v, v\rangle &= \frac{\kappa_{\infty}}{\lambda} \bigg(\langle \mathbf{1}v,v\rangle+ \langle \boldsymbol{G}^* v,v\rangle - e^{-\lambda T} \int_0^T u(t) \dd t \int_0^T e^{\lambda t} u(t) \dd t \bigg)\\
    &= \frac{\kappa_{\infty}}{\lambda} \bigg(\langle \mathbf{1}v,v\rangle+ \langle \boldsymbol{G} v,v\rangle - e^{-\lambda T} \int_0^T u(t) \dd t \int_0^T e^{\lambda t} u(t) \dd t \bigg).
\end{align*}
We define
\begin{equation*}
    y(t) := \boldsymbol{G}v(t) = \int_0^t e^{-\lambda (t-s)} v(s) \dd s \quad \text{and} \quad z(t) := \mathbf{1}v(t) = \int_0^t v(s) \dd s,
\end{equation*}
so that $y{'} + \lambda y = v$ and $z' = v$. We then write
\begin{equation*}
    \langle \boldsymbol{V}v, v \rangle = \frac{\kappa_{\infty}}{\lambda} \bigg(\int_0^T z(t)v(t) \dd t + \int_0^T y(t)v(t) \dd t     - z(T)y(T)\bigg).
\end{equation*}
Furthermore, we have
\begin{align*}
    \int_0^T z(t)v(t) \dd t &= \int_0^T z(t)z'(t) \dd t = \frac{1}{2}z(T)^2,\\
    \int_0^T y(t)v(t)\dd t &= \int_0^T y'(t)y(t) \dd t + \lambda \int_0^T y(t)^2 \dd t = \frac{1}{2} y(T)^2 + \lambda \int_0^T y(t)^2 \dd t.
\end{align*}
Combining these identities, we get
\begin{align*}
    \langle \boldsymbol{V}v,v\rangle & = \frac{\kappa_{\infty}}{\lambda}\bigg(\frac{1}{2} z(T)^2 + \frac{1}{2}y(T)^2 + \lambda \int_0^T y(t)^2 \dd t -z(T)y(T)\bigg) \\
    & = \frac{\kappa_{\infty}}{2\lambda} \bigg(\big(z(T) - y(T)\big)^2 + 2 \lambda \int_0^T y(t)^2 \dd t \bigg) \ge 0
\end{align*}
which proves~\eqref{eq:goal:24} when  $G = G_\lambda$.

\begin{comment}
As a by-product, considering $L^{2}(\mathbb{R})$ instead of $L^{2}([0,T])$, there is no way to extend the operator $\boldsymbol{G_{\lambda}}:L^{2}([0,T]) \to L^{2}([0,T])$ to $\boldsymbol{\tilde{G}_{\lambda}}: L^{2}(\mathbb{R}) \to L^{2}(\mathbb{R})$ such that $\boldsymbol{\tilde{G}_{\lambda}} |_{L^{2}([0,T])} = \boldsymbol{G_{\lambda}}$.

\paragraph{General case.} Any completely monotone kernel can be written as the Laplace transform of a positive measure $\mu$ such that
\begin{equation*}
    G(t) = \int_0^\infty e^{-\lambda t} \, \mu(d \lambda),
\end{equation*}
see for example (ADD REF). Therefore, using Fubini, we readily obtain from the previous computations
\begin{equation*}
    \boldsymbol{K} = \int_0^\infty \frac{1}{2\lambda} ( ( \boldsymbol{H} + \boldsymbol{H}^{*} ) + \kappa_{\infty} ( \boldsymbol{G_{\lambda}} + \boldsymbol{G_{\lambda}}^{*} ) ) \mu(d \lambda),
\end{equation*}
which is indeed positive semi-definite as an infinite sum of psd operators... ANY REF FOR THIS ONE?
\end{comment}

\paragraph{Proof in the general case.} First recall that since $G$ is a completely monotone kernel in $L^2$, there exists a $\sigma$-finite measure $\mu$ such that
\begin{equation*}
    G(t) = \int_0^\infty e^{-\lambda t} \mu(\dd \lambda), \quad t \ge 0,
\end{equation*}
see Theorem \ref{thm:bernstein_widder}. When $\mu$ is finite, we can apply Fubini's theorem which ensures that
\begin{equation*}
(\boldsymbol{G}v)(t) = \int_0^t \bigg(\int_{0}^{\infty} e^{-\lambda s} \mu(\dd\lambda)\bigg) v(t-s) \dd s 
= \int_{0}^{\infty} (\boldsymbol{G}_\lambda v)(t) \, \mu(\dd \lambda),
\end{equation*}
where $\boldsymbol{G}_\lambda$ is the operator whose kernel is given by~\eqref{eq:exp_kernel_for_proof}.
Applying Fubini's theorem again gives
\begin{equation*}
\langle \boldsymbol{H}v, \boldsymbol{G}v\rangle 
=
\int_{0}^{\infty} \langle \boldsymbol{H}v, \boldsymbol{G}_\lambda v\rangle \, \mu(\dd \lambda).
\end{equation*}
We've already proved that $\langle \boldsymbol{H}v, \boldsymbol{G}_\lambda v\rangle \geq 0$ for all $\lambda > 0$ and therefore~\eqref{eq:goal:24} holds. \\

When $\mu$ is infinite, we first consider $(A_n)_n$ an increasing sequence of subsets on $[0, \infty)$ such that $A_n \to [0, \infty)$ and $\mu(A_n) < \infty$. We then set $\mu_n = \mu \1_{A_n}$ and 
\begin{equation*}
G^{(n)}(t) = \int_0^\infty e^{-t\lambda} \, \mu_n(\dd \lambda).
\end{equation*}
This kernel is an approximation of $G$. In fact, we have
\begin{equation*}
     0 \leq G^{(n)}(t) \leq G(t)
     \qquad \text{ and } \qquad
     \lim\limits_{n \to \infty} G^{(n)}(t) = G(t)
\end{equation*}
for all $t \ge 0$.
Since $G$ is in $L^2$, the dominated convergence theorem applies and we have
\begin{equation}
\label{eq:conv:Gn}
\int_0^T \big(G^{(n)}(s) - G(s)\big)^2 \dd s \to 0.
\end{equation}
Moreover, $G^{(n)}$ is a completely monotone kernel and in $L^2$ because it is bounded by $G$. Thus, as we have already proved, we know that $\langle \boldsymbol{H}v, \boldsymbol{G}^{(n)} v\rangle \geq 0$ where $\boldsymbol{G}^{(n)}$ is the convolution operator associated with $G^{(n)}$. To conclude, it remains to show that $\boldsymbol{G^{(n)}} v \to \boldsymbol{G} v$ in $L^2([0,T])$. In fact, we have
\begin{align*}
\norm{\boldsymbol{G^{(n)}} v - \boldsymbol{G} v}^2
 &=
 \int_0^T
 |
 (\boldsymbol{G^{(n)}} v - \boldsymbol{G} v)(t)|^2
  \dd t
 \\
 &= 
 \int_0^T
 \bigg| \int_0^t v(t-s)\big(G^{(n)}(s) - G(s)\big) \dd s \bigg|^2
 \dd t
 \\
 &\leq
 T
 \norm{u}^2 \int_0^T \big(G^{(n)}(s) - G(s)\big)^2 \dd s
 \end{align*}
 which converges to $0$ by~\eqref{eq:conv:Gn}.

\end{proof}

\section{Proof of Theorem~\ref{thm:existence}}

\subsection{Outline of the proof}

Since $\mathcal{L}^2$ is a Hilbert space, an application of \citet[Theorem 1.2, Chapter 1]{struwe2000variational} ensures the existence of a maximiser of $\mathcal{J}$ provided $-\mathcal{J}$ is coercive and weakly lower semi-continuous. Therefore, we can conclude using the following two lemmas.

\begin{lemma}
    \label{lemma:wslc}
    Suppose that $\Omega$ is countable or finite and that Assumptions~\ref{assumption:H},~\ref{assumption:G} and~\ref{assumption:U} hold. Then $\calJ$ is weakly continuous in $L^2$.
\end{lemma}

\begin{lemma}
\label{lemma:coercive}
    Suppose that $\kappa_{\infty} + \gamma > 0$ and that Assumptions~\ref{assumption:H},~\ref{assumption:G} and~\ref{assumption:U} hold. Then $-\calJ$ is coercive in $L^2$.
\end{lemma}

The proof of these two lemmas is deferred to Sections~\ref{sec:proof:lemma:coercive} and~\ref{sec:proof:lemma:wslc}. They rely on the properties of the operator $\mathbf{B}$ defined on $\mathcal{L}^2$ by
\begin{equation*}
    \mathbf{B} u(t) = \calU\big(\mathbf{G} u(t)\big), \quad t \geq 0.
\end{equation*}
The properties of this operator are studied in Section~\ref{sec:prop:B}.

\subsection{Properties of the operator $\mathbf{B}$}
\label{sec:prop:B}

\begin{lemma}
\label{lemma:IRIA}
We have $\mathbf{I} - \mathbf{R} = (\mathbf{I} + \mathbf{B})^{-1}$. Moreover, there exists $c > 0$ such that
    \begin{equation}
    \label{eq:control:IpB}
        c^{-1}\norm{u} \leq \norm{(\mathbf{I} + \mathbf{B})^{-1}} \leq c \norm{u}.
    \end{equation}
\end{lemma}

\begin{proof}
    Note that for all $u \in \mathcal{L}^2$, $w = (\mathbf{I} - \mathbf{R})u$ satisfies 
    \begin{equation*}
    w(t) = u(t) - r^u(t) = u(t) - \calU\big(\mathbf{G}w(t)\big)
    \end{equation*}
    and thus $
    u = (\mathbf{I} + \mathbf{B})w$. Moreover, adapting the results of Appendix~\ref{app:nonlinear_volterra}, we see that for each $u$, the equation 
    \begin{equation*}
            w(t) + \calU\big(\mathbf{G}w(t)\big)= u(t) 
    \end{equation*}
    admits a unique solution in $\mathcal{L}^2$. This ensures that $\mathbf{I} + \mathbf{B}$ is invertible and thus $\mathbf{I} - \mathbf{R} = (\mathbf{I} + \mathbf{B})^{-1}$.\\

    To prove~\eqref{eq:control:IpB}, we use the fact that both $\mathbf{I}+\mathbf{B} $ and $(\mathbf{I}+\mathbf{B})^{-1}$ are Lipschitz, which follows readily from the fact that $\mathbf{B}$ is Lipschitz continuous because both $\mathcal{U}$ and $\mathbf{G}$ are Lipschitz continuous and from the fact that $\mathbf{R}$ is Lipschitz continuous.
\end{proof}

Note that the operator $\mathbf{B}$ satisfies the assumptions of Lemma~\ref{lem:calL2_to_L2}: it is defined through a deterministic relation and thus can naturally extend to a non-anticipative operator $\widecheck{\mathbf{B}}$ in $L^2$. The same holds for $\mathbf{R}$ and we also have 
    \begin{equation*}
        \mathbf{I} - \widecheck{\mathbf{R}} = (\mathbf{I} + \widecheck{\mathbf{B}})^{-1}.
    \end{equation*}
The properties of $\widecheck{\mathbf{B}}$ are studied below.

\begin{lemma}
The operator $\widecheck{\mathbf{B}}$ is weakly continuous in $L^2$.
\end{lemma}
\begin{proof}
Let $x_n \rightharpoonup x$. Since $\mathbf{G}_{1}$ is a convolution operator, it is also compact and therefore, we know that $\mathbf{G}_{1}x_n \to \mathbf{G}_{1}x$. Then since $\calU$ is Lipschitz continuous, we deduce that $\widecheck{\mathbf{B}} x_n \to \widecheck{\mathbf{B}} x$ which implies that $\widecheck{\mathbf{B}} x_n \rightharpoonup \widecheck{\mathbf{B}} x$.
\end{proof}

\begin{lemma}
The operator $(\mathbf{I} + \widecheck{\mathbf{B}})^{-1}$ is weakly continuous in $L^2$.
\end{lemma}
\begin{proof}
We consider $y_n \rightharpoonup y$ and we define
\begin{equation*}
x_n = (\mathbf{I} + \widecheck{\mathbf{B}})^{-1}y_n \quad \text{ and } \quad x = (\mathbf{I} + \widecheck{\mathbf{B}})^{-1}y
\end{equation*}
so that
\begin{equation*}
y_n = x_n + \widecheck{\mathbf{B}}x_n \quad \text{ and } \quad y = x + \widecheck{\mathbf{B}}x.
\end{equation*}

We want to prove that $x_n \rightharpoonup x$. The proof is done in two steps: first we show that the sequence $(x_n)$ lies in a weak compact set, and then show that if $x_{n_k} \rightharpoonup x'$ for some $n_k$, then we must have $x' = x$. \\

\textbf{Step 1: $(x_n)$ lies in a weak compact set.}
We start by showing that $(x_n)_n$ is bounded in $L^2$. As a start, we define
\begin{equation*}
r_n = y_n - x_n = y_n - (\mathbf{I} + \widecheck{\mathbf{B}})^{-1}y_n = \mathbf{R} (y_n).
\end{equation*}
We know from Lemma~\ref{lem:well_posed} that $\mathbf{R}$ is Lipschitz continuous, and therefore there exists $C > 0$ such that 
\begin{equation*}
    \norm{r_n} \leq C \norm{y_n}.
\end{equation*}
Therefore, we have
\begin{equation*}
    \norm{x_n} = \norm{y_n - r_n} \leq (C+1) \norm{y_n}.
\end{equation*}
But $(y_n)_n$ is a bounded sequence in $L^2$ because it is weakly convergent and therefore $(x_n)_n$ is bounded. This implies in turn that $(x_n)_n$ lies in a weak compact set because every bounded and closed set in a Hilbert space is weakly relatively compact.\\

\textbf{Step 2: Identification of the weak accumulation point of $(x_n)_n$.} Consider a weakly convergent subsequence of $(x_n)$, that we still write $(x_n)$ for conciseness. We denote by $x'$ its limit and we prove that $x' = x$. By weak continuity of $\widecheck{\mathbf{B}}$, we know that
\begin{equation*}
\widecheck{\mathbf{B}}x_n \rightharpoonup \widecheck{\mathbf{B}}x'
\end{equation*}
and thus
\begin{equation*}
y_n = x_n + \widecheck{\mathbf{B}}x_n \rightharpoonup x' + \widecheck{\mathbf{B}}x'.
\end{equation*}
By uniqueness of the limit, we have $y = x' + \widecheck{\mathbf{B}}x'$ which implies that $x' = (\mathbf{I} + \widecheck{\mathbf{B}})^{-1}y = x$.
\end{proof}

\subsection{Proof of Lemma~\ref{lemma:wslc}}
\label{sec:proof:lemma:wslc}
From~\eqref{eq:operator:J}, we know that
\begin{equation*}
-\mathcal{J}(u) 
    = 
    -\langle u, \alpha \rangle + \frac{\gamma}{2} \norm{u}^2 + \langle u, (\mathbf{H} + \mathbf{G})\circ(\mathbf{I} - \mathbf{R})u \rangle
    + \frac{\phi}{2} \norm{\bfX u}^{2}
    + \frac{\varrho}{2} \E[(\bfX u)_T^2] - X_0 \mathbb{E}[S_T].
\end{equation*}
Following \citet[Lemma 6.2]{abi2025fredholm}, we know that $u \mapsto -\langle u, \alpha \rangle + \frac{\gamma}{2} \norm{u}^2$ and
$u \mapsto \phi \norm{\bfX u}^{2}
    + \varrho \E[(\bfX u)_T^2] $
are weakly lower semi-continuous. It remains to study $u \mapsto \langle u, (\mathbf{H} + \mathbf{G}) \circ(\mathbf{I} - \mathbf{R})u \rangle$. Lemma~\ref{lemma:IRIA} ensures that $(\mathbf{H} + \mathbf{G})\circ (\mathbf{I} - \mathbf{R}) = (\mathbf{H} + \mathbf{G}) \circ (\mathbf{I} + \mathbf{B})^{-1}$. This is done in two steps:
\begin{itemize}
    \item \textbf{Step 1.} We first study $(\widecheck{\mathbf{H}} + \widecheck{\mathbf{G}}) \circ (\mathbf{I} + \widecheck{\mathbf{B}})^{-1}$ in $L^2$ and we show that $u \mapsto \langle u, (\widecheck{\mathbf{H}} + \widecheck{\mathbf{G}}) \circ (\mathbf{I} - \widecheck{\mathbf{R}})u \rangle$ is weakly semi-continuous in $L^2$.
    \item \textbf{Step 2.} We then show that weak lower semi-continuity extends to $\mathcal{L}^2$ using crucially that $\Omega$ is finite or countable.
\end{itemize}

\textbf{Proof of Step 1.} Recall first that although the operators $\widecheck{\mathbf{H}}$, $\widecheck{\mathbf{G}}$ and $\widecheck{\mathbf{B}}$ are defined on $\mathcal{L}^2$, they can be seen as operator in $L^2$ because all these operators are defined through deterministic relations. We are now ready to prove that $u \mapsto \langle(\widecheck{\mathbf{H}} + \widecheck{\mathbf{G}}) \circ(\widecheck{\mathbf{I}} + \widecheck{\mathbf{B}})^{-1} u, u\rangle$ is weakly continuous in $L^2$. Consider $u_n \rightharpoonup u$. We want to prove that 
\begin{equation*}
\langle(\widecheck{\mathbf{H}} + \widecheck{\mathbf{G}}) \circ(\widecheck{\mathbf{I}} + \widecheck{\mathbf{B}})^{-1} u_n, u_n\rangle \to \langle(\widecheck{\mathbf{H}} + \widecheck{\mathbf{G}}) \circ(\widecheck{\mathbf{I}} + \widecheck{\mathbf{B}})^{-1} u, u\rangle.
\end{equation*}
We have
\begin{equation*}
|\langle(\widecheck{\mathbf{H}} + \widecheck{\mathbf{G}}) \circ(\widecheck{\mathbf{I}} + \widecheck{\mathbf{B}})^{-1} u_n, u_n\rangle - \langle( \widecheck{\mathbf{H}} + \widecheck{\mathbf{G}}) \circ(\widecheck{\mathbf{I}} + \widecheck{\mathbf{B}})^{-1} u, u\rangle| \leq \mathrm{I} + \mathrm{II} 
\end{equation*}
where
\begin{align*}
\mathrm{I}  &= |\langle(\widecheck{\mathbf{H}} + \widecheck{\mathbf{G}}) \circ(\widecheck{\mathbf{I}} + \widecheck{\mathbf{B}})^{-1} u_n - (\widecheck{\mathbf{H}} + \widecheck{\mathbf{G}}) \circ (\widecheck{\mathbf{I}} + \widecheck{\mathbf{B}})^{-1} u, u_n\rangle|
\\
\mathrm{II}  &= |\langle(\widecheck{\mathbf{H}} + \widecheck{\mathbf{G}})\circ (\widecheck{\mathbf{I}} + \widecheck{\mathbf{B}})^{-1} u, u_n - u\rangle|.
\end{align*}
Clearly $\mathrm{II} \to 0$ because $u_n \rightharpoonup u$ and $(\widecheck{\mathbf{H}} + \widecheck{\mathbf{G}}) \circ (\widecheck{\mathbf{I}} + \widecheck{\mathbf{B}})^{-1} u$ is in $L^2$. Moreover, using the weak continuity of $(\widecheck{\mathbf{I}} + \widecheck{\mathbf{B}})^{-1}$, we have
\begin{equation*}
(\widecheck{\mathbf{I}} + \widecheck{\mathbf{B}})^{-1} u_n \rightharpoonup (\widecheck{\mathbf{I}} + \widecheck{\mathbf{B}})^{-1} u.
\end{equation*}
Seeing that $\widecheck{\mathbf{H}} + \widecheck{\mathbf{G}}$ is a compact operator because it is a convolution operator, we deduce that
\begin{equation*}
(\widecheck{\mathbf{H}} + \widecheck{\mathbf{G}}) \circ (\widecheck{\mathbf{I}} + \widecheck{\mathbf{B}})^{-1} u_n \to  (\widecheck{\mathbf{H}} + \widecheck{\mathbf{G}}) \circ (\widecheck{\mathbf{I}} + \widecheck{\mathbf{B}})^{-1} u.
\end{equation*}
In addition, $(u_n)_n$ is bounded because it is weakly convergent, and thus
\begin{equation*}
I \leq \norm{(\widecheck{\mathbf{H}} + \widecheck{\mathbf{G}}) \circ (\widecheck{\mathbf{I}} + \widecheck{\mathbf{B}})^{-1} u_n - (\widecheck{\mathbf{H}} + \widecheck{\mathbf{G}}) \circ (\widecheck{\mathbf{I}} + \widecheck{\mathbf{B}})^{-1}  u} \norm{u_n} \to 0.
\end{equation*}

\textbf{Proof of Step 2.} We now prove that the weak lower semi-continuity extends to $\mathcal{L}^2$. Consider $(u_n)_n$ a sequence of $\mathcal{L}^2$ processes converging weakly to $u$. Following the proof of \citet[Lemma 6.2]{abi2025fredholm}, we see that the for $\mathbb{P}$-almost all $\omega$, we have
\begin{equation*}
    u_n(\omega, \cdot) \rightharpoonup u(\omega, \cdot) \qquad \text{ in } L^{2}.
\end{equation*}
Therefore, applying the results of Step 1, we have, for $\mathbb{P}$-almost all $\omega$
\begin{equation*}
\langle(\widecheck{\mathbf{H}} + \widecheck{\mathbf{G}}) \circ (\widecheck{\mathbf{I}} + \widecheck{\mathbf{B}})^{-1} u_n(\omega, \cdot), u_n(\omega, \cdot) \rangle_{L^2([0,T])} \to \langle(\widecheck{\mathbf{H}} + \widecheck{\mathbf{G}}) \circ (\widecheck{\mathbf{I}} + \widecheck{\mathbf{B}})^{-1} u(\omega, \cdot), u(\omega, \cdot)\rangle_{L^2([0,T])}
\end{equation*}
as $n \to \infty$. Using Fatou's Lemma, we get
\begin{equation*}
    \liminf_{n \to \infty} \E [ \langle(\widecheck{\mathbf{H}} + \widecheck{\mathbf{G}}) \circ (\widecheck{\mathbf{I}} + \widecheck{\mathbf{B}})^{-1} u_n(\omega, \cdot), u_n(\omega, \cdot) \rangle_{L^2([0,T])} 
    ] 
    \geq
    \E [ \langle(\widecheck{\mathbf{H}} + \widecheck{\mathbf{G}}) \circ (\widecheck{\mathbf{I}} + \widecheck{\mathbf{B}})^{-1} u(\omega, \cdot), u(\omega, \cdot)\rangle_{L^2([0,T])} ].
\end{equation*}
Using the definition of the scalar product in $\calL^2$, we deduce that
\begin{equation*}
    \liminf_{n \to \infty} 
    \,
    \langle({\mathbf{H}} + {\mathbf{G}}) \circ ({\mathbf{I}} + {\mathbf{B}})^{-1} u_n, u_n \rangle_{\calL^2([0,T])} 
    \geq
    \langle({\mathbf{H}} + {\mathbf{G}}) \circ({\mathbf{I}} + {\mathbf{B}})^{-1} u, u\rangle_{\calL^2([0,T])},
\end{equation*}
which exactly entails the weak lower semi-continuity of $u \mapsto \langle u, (\mathbf{H} + \mathbf{G}) \circ (\mathbf{I} - \mathbf{R})u \rangle$.

\subsection{Proof of Lemma~\ref{lemma:coercive}}
\label{sec:proof:lemma:coercive}

From~\eqref{eq:operator:J}, we know that
\begin{align*}
-\mathcal{J}(u) 
    &= 
    -\langle u, \alpha \rangle + \frac{\gamma}{2} \norm{u}^2 + \langle u, (\mathbf{H} + \mathbf{G})\circ(\mathbf{I} + \mathbf{B})^{-1} u \rangle
    + \frac{\phi}{2} \norm{\bfX u}^{2}
    + \frac{\varrho}{2} \E[(\bfX u)_T^2] + X_0 \mathbb{E}[S_T]
\\
    &\geq
    -\langle u, \alpha \rangle + \frac{\gamma}{2} \norm{u}^2 + \langle u, (\mathbf{H} + \mathbf{G})\circ(\mathbf{I} + \mathbf{B})^{-1} u \rangle
    + X_0 \mathbb{E}[S_T].
\end{align*}
To prove the coercivity of $-\mathcal{J}$, we first show the following bound.

\begin{lemma}
\label{lemma:coercive:1}
    For each $\epsilon > 0$, there exists $b > 0$ such that 
    \begin{equation*}
    \langle u, (\mathbf{H} + \mathbf{G}) \circ (\mathbf{I} + \mathbf{B})^{-1} u \rangle
    \geq 
    b +
    \Big( \frac{\kappa_{\infty}}{2} - \epsilon \Big) \norm{u}^2.
    \end{equation*}
\end{lemma}

\begin{proof}[Proof of Lemma~\ref{lemma:coercive:1}]

    We fix $u \in \mathcal{L}^2$ and define $v = (\mathbf{I} + \mathbf{B})^{-1} u$ so that 
\begin{equation*}
\langle u, (\mathbf{H} + \mathbf{G}) \circ (\mathbf{I} + \mathbf{B})^{-1} u \rangle
=
\langle v + \mathbf{B} v, (\mathbf{H} + \mathbf{G}) v \rangle
=
\langle (\mathbf{H} + \mathbf{G})v, v \rangle 
+ \langle (\mathbf{H} + \mathbf{G})v, \mathbf{B} v \rangle.
\end{equation*}
We then define $\wt U(x) = U(x) - \delta x$
where $\delta$ is defined in Assumption~\ref{assumption:U}. Then $\wt U$ is a Lipschitz continuous function and
\begin{equation*}
\frac{\wt U(x)}{x} \to 0
\end{equation*}
when $x \to \infty$. This implies in particular that for all $\epsilon > 0$, there exists $c > 0$ such that
$|\wt U(x)| \leq c + \epsilon |x|$ for all $x \in \mathbb{R}$. We now define the operator 
\begin{equation*}
\widetilde{\mathbf{B}}w = \wt U ( \mathbf{G} w )
\end{equation*}
so that $\mathbf{B} = \widetilde{\mathbf{B}} + \delta \mathbf{G}$. We then write
\begin{align*}
\langle u, (\mathbf{H} + \mathbf{G}) \circ (\mathbf{I} + \mathbf{B})^{-1} u \rangle
&=
\langle (\mathbf{H} + \mathbf{G})v, v \rangle 
+ \delta \langle (\mathbf{H} + \mathbf{G})v, \mathbf{G} v \rangle
+ \langle (\mathbf{H} + \mathbf{G})v, \widetilde{\mathbf{B}} v \rangle
\\
&=
\langle \mathbf{H}v, v \rangle 
+
\langle \mathbf{G} v, v \rangle 
+ \delta \langle \mathbf{H} v, \mathbf{G} v \rangle
+ \delta \norm{\mathbf{G} v}^2
+ \langle (\mathbf{H} + \mathbf{G})v, \widetilde{\mathbf{B}} v \rangle.
\end{align*}
Using the specification of $\mathbf{H}$ in Assumption~\ref{assumption:H}, we have that 
\begin{equation*}
    \langle \mathbf{H}v, v \rangle = \frac{\kappa_{\infty}}{2} \norm{v}^2.
\end{equation*}
Using~\eqref{eq:goal:24} and the fact that $\mathbf{G}$ is semi-definite positive, we deduce that
\begin{equation*}
    \langle u, (\mathbf{H} + \mathbf{G}) \circ (\mathbf{I} + \mathbf{B})^{-1} u \rangle
    \geq 
    \frac{\kappa_{\infty}}{2} \norm{v}^2
    + \langle (\mathbf{H} + \mathbf{G})v, \widetilde{\mathbf{B}} v \rangle.
\end{equation*}
For the last term, we have
\begin{align*}
    | \langle (\mathbf{H} + \mathbf{G})v, \widetilde{\mathbf{B}} v \rangle|
    & \leq
    \E \bigg[
    \int_0^T 
    |(\mathbf{H} + \mathbf{G})v(t) 
    \widetilde{\mathbf{B}} v(t)| \mathrm{d} t
    \bigg]
\\
    & \leq
    c
    \E \bigg[
    \int_0^T 
    |(\mathbf{H} + \mathbf{G})v(t) 
    |  \mathrm{d} t
    \bigg]
+
    \epsilon
    \E \bigg[
    \int_0^T 
    |(\mathbf{H} + \mathbf{G})v(t) 
    \widetilde{\mathbf{G}} v(t)|  \mathrm{d} t
    \bigg]
\\
    & \leq
    c
    \E \bigg[
    \int_0^T 
    (\mathbf{H} + \mathbf{G})|v|(t) 
    \mathrm{d} t
    \bigg]
+
    \epsilon
    \E \Big[
    \int_0^T 
    (\mathbf{H} + \mathbf{G})|v|(t) 
    \widetilde{\mathbf{G}} |v|(t) \, \mathrm{d} t
    \Big]
\end{align*}
From the continuity properties of $\mathbf{H}$ and $\mathbf{G}$, we deduce that there exist $c_1, c_2 > 0$ such that
\begin{equation*}
    | \langle (\mathbf{H} + \mathbf{G})v, \widetilde{\mathbf{B}} v \rangle |
    \leq
    c_1 \norm{v} + c_2 \epsilon \norm{v}^2
\end{equation*}
and therefore
\begin{equation*}
    \langle u, (\mathbf{H} + \mathbf{G}) \circ (\mathbf{I} + \mathbf{B})^{-1} u \rangle
    \geq 
    \Big( \frac{\kappa_{\infty}}{2} - c_2 \epsilon \Big) \norm{v}^2
    - c_1 \norm{v}.
\end{equation*}
We conclude using~\eqref{eq:control:IpB}.
\end{proof}

It remains to show how Lemma~\ref{lemma:coercive:1} implies Lemma~\ref{lemma:coercive}. Using that $u$ and $\alpha$ are both in $\mathcal{L}^2$, then for each $\epsilon > 0$, there exists $b > 0$ such that
\begin{equation*}
    -\mathcal{J}(u) \geq b + \Big( \frac{\kappa_{\infty} + \gamma}{2} - \epsilon \Big) \norm{u}^2.
\end{equation*}
We can conclude whenever $\kappa_{\infty} + \gamma > 0$ by taking $0 < \epsilon < \kappa_{\infty} + \gamma$.

\end{document}